\documentclass[onecolumn,11pt,draftcls]{IEEEtran}
\usepackage{amsfonts,amssymb,amsmath,amsthm}
\usepackage{algorithmic,algorithm}
\usepackage{xcolor}
\usepackage{ifpdf}
\ifpdf
\usepackage[]{graphicx}
\else
\usepackage[]{graphicx}
\fi
\usepackage{url}
\usepackage{cite}
\usepackage{changepage}
\newtheorem{theorem}{Theorem}
\newtheorem{lemma}{Lemma}
\newtheorem{definition}{Definition}
\newtheorem{example}{Example}
\newtheorem{algorithms}{Algorithm}
\def\calE{{\cal E}}
\def\calH{{\cal H}}

\def\calF{{\cal F}}

\def\calO{{\cal O}}

\def\calG{{\cal G}}

\newenvironment{namelist}[1]{%
\begin{list}{}
  {
   
   \settowidth{\labelwidth}{#1}
   \setlength{\leftmargin}{1.1\labelwidth}
   }
  }{%
\end{list}}
\newcommand{\step}[1]{\textbf{Step {#1}:}}
\def\aa{\hspace{5mm}}

\usepackage[normalem]{ulem}

\begin{document}
\title{Beyond the MDS Bound in Distributed Cloud Storage}
\author{Jian Li\aa Tongtong Li\aa Jian Ren\footnote{The authors are with the Department of ECE, Michigan State University, East Lansing, MI 48824-1226.  Email: \{lijian6, tongli, renjian\}@msu.edu. 

This paper was presented in part at IEEE INFOCOM 2014.}}

\date{May 28, 2014}

\maketitle

\begin{abstract}
Distributed storage plays a crucial role in the current cloud computing framework. After the theoretical bound for distributed storage was derived by the pioneer work of the regenerating code, Reed-Solomon code based regenerating codes were developed.  The RS code based minimum storage regeneration code (RS-MSR) and the minimum bandwidth regeneration code (RS-MBR) can achieve theoretical bounds on the MSR point and the MBR point respectively in code regeneration.  They can also maintain the MDS property in code reconstruction.
However, in the hostile network where the storage nodes can be compromised and the packets can be tampered with, the storage capacity of the network can be significantly affected.
In this paper, we propose a Hermitian code based minimum storage regenerating~(H-MSR) code and a minimum bandwidth regenerating~(H-MBR) code. We first prove that our proposed Hermitian code based regenerating codes can achieve the theoretical bounds for MSR point and MBR point respectively.
We then propose data regeneration and reconstruction algorithms for the H-MSR code and the H-MBR code in both error-free network and hostile network.
Theoretical evaluation shows that our proposed schemes can detect the erroneous decodings and correct more errors in hostile network than the RS-MSR code and the RS-MBR code with the same code rate.  Our analysis also demonstrates that the proposed H-MSR and H-MBR codes have lower computational complexity than the RS-MSR/RS-MBR codes in both code regeneration and code reconstruction.

\end{abstract}


\section{Introduction}

Cloud storage is an on-demand network data storage and access paradigm.  To ensure accessibility of remotely stored data at any time, a typical solution is to store data across multiple servers or clouds, often in a replicated fashion. Data replication not only lacks flexibility in data recovery, but also requires secure data management and protection.

It is well known that secure data management is generally very costly and may be vulnerable to data compromising attacks.
Distributed data storage provides an elegant tradeoff between the costly secure data management task and the cheap storage media.  The main idea for distributed storage is that instead of storing the entire data in one server, we can split the data into $n$ data components and store the components across multiple servers. The original data can be recovered only when the required (threshold) number of components, say $k$, are collected.
In fact, the original data is information theoretically secure for anyone who has access up to $k-1$ data components.
In this case, when the individual components are stored distributively across multiple cloud storage servers, each cloud storage server only needs to ensure data integrity and data availability.  The costly data encryption and secure key management may no longer be needed any more.  The distributed cloud storage can also increase data availability while reducing network congestion that leads to increased resiliency.  A popular approach is to employ an $(n,k)$ maximum distance separable (MDS) code such as an Reed-Solomon (RS) code \cite{Ocean,Total}.  For RS code, the data is stored in $n$ storage nodes in the network. The data collector~(DC) can reconstruct the data by connecting to any $k$ healthy nodes.

While RS code works perfect in reconstructing the data, it lacks scalability in repairing or regenerating a failed node.  To deal with this issue, the concept of regenerating code was introduced in \cite{Dimakis}.  The main idea of the regenerating code is to allow a replacement node to connect to some individual nodes directly and regenerate a substitute of the failed node, instead of first recovering the original data then regenerating the failed component.

Compared to the RS code, regenerating code achieves an optimal tradeoff between bandwidth and storage within the minimum storage regeneration (MSR) and the minimum bandwidth regeneration (MBR) points.  RS code based MSR (RS-MSR) code and MBR (RS-MBR) code have been explicitly constructed in \cite{Rashmi}.
However, the existing research either has no error detection capability, or has the error correction capability limited by the RS code.
Moreover, the schemes with error correction capability are unable to determine whether the error correction is successful.

In this paper, we construct the regenerating codes by combining the Hermitian code and regenerating code at the MSR point (H-MSR code) and the MBR point (H-MBR code).
We prove that these codes can achieve the theoretical MSR bound and MBR bound respectively. We also propose data regeneration and reconstruction algorithms for the H-MSR code and the H-MBR code in both error-free network and hostile network. Our proposed algorithms can detect the errors in the network while achieving the error correction capability \emph{beyond the RS code}.  Moreover, our proposed algorithms can determine whether the error correction is successful.

The rest of this paper is organized as follows:
In Section \ref{Sec:RelatedWork}, related work is reviewed.
The preliminary of this paper is presented in Section \ref{Sec:Preliminary}.  In Section \ref{Sec:H-MSR}, our proposed encoding of  H-MSR code is described.   In Section \ref{Sec:H-MSR-Regeneration}, regeneration of the H-MSR code is discussed.  Reconstruction of the H-MSR code is analyzed in Section \ref{Sec:H-MSR-Reconstruct}.
In Section \ref{Sec:H-MBR}, our proposed encoding of  H-MBR code is described.   In Section \ref{Sec:H-MBR-Regeneration}, regeneration of the H-MBR code is discussed.  Reconstruction of the H-MBR code is analyzed in Section \ref{Sec:H-MBR-Reconstruct}.   We conduct performance analysis in Section \ref{Sec:Performance}.  The paper is concluded in Section \ref{Sec:Conclusion}.

\section{Related Work} \label{Sec:RelatedWork}

When a storage node in the distributed storage network that employing the conventional $(n,k)$ RS code (such as OceanStore~\cite{Ocean} and Total Recall~\cite{Total}) fails,
the replacement node connects to $k$ nodes and downloads the whole file to recover the symbols stored in the failed node. This approach is a waste of bandwidth because the whole file has to be downloaded to recover a fraction of it.
To overcome this drawback, Dimakis \emph{et al.} ~\cite{Dimakis} introduced the concept of $\{n, k, d, \alpha, \beta, B\}$ regenerating code. In the context of regenerating code, the replacement node can regenerate the contents stored in a failed node by downloading $\gamma$ help symbols from $d$ helper nodes. The bandwidth consumption to regenerate a failed node could be far less than the whole file. A data collector~(DC)
can reconstruct the original file stored in the network by downloading $\alpha$ symbols from each of the $k$ storage nodes.
In~\cite{Dimakis}, the authors proved that there is a tradeoff between bandwidth $\gamma$ and per node storage $\alpha$.  They find two optimal points: minimum storage regeneration (MSR) and minimum bandwidth regeneration (MBR) points.

The regenerating code can be divided into functional regeneration and exact regeneration. In the functional regeneration, the replacement node
regenerates a new component that can functionally replace the failed component instead of being the same as the original stored component.
\cite{Ywu} formulated the data regeneration as a multicast network coding problem and constructed functional regenerating codes.
\cite{Duminuco} implemented a random linear regenerating codes
for distributed storage systems.
\cite{Shum} proved that by allowing data exchange among the replacement nodes, a better tradeoff between repair bandwidth $\gamma$ and per node storage $\alpha$ can be achieved.
In the exact regeneration, the replacement node regenerates the exact symbols of a failed node.
\cite{Ywu2} proposed to reduce the regeneration bandwidth through algebraic alignment.
\cite{Shah} provided a code structure for exact regeneration using interference alignment technique.
\cite{Rashmi} presented optimal exact constructions of MBR codes and MSR codes under product-matrix framework. This is the first work that allows independent selection of the nodes number $n$ in the network.

However, none of the existing work above considered code regeneration under node corruption or adversarial manipulation attacks. In fact, all these schemes will fail in both regeneration and reconstruction in these scenarios.

In~\cite{Pawar}, the authors discussed the amount of information that can be safely stored against passive eavesdropping and active adversarial attacks based on the regeneration structure.
In~\cite{Han}, the authors proposed to add CRC codes in the regenerating code to check the integrity of the data in hostile network. Unfortunately, the CRC checks can also be manipulated by the malicious nodes, resulting in the failure of the regeneration and reconstruction.
In \cite{Rashmi-err}, the authors analyzed the error resilience of the RS code based regenerating code in the network with errors and erasures. They provided the theoretical error correction capability.  Their result is an extension of the MDS code to the regenerating code and their scheme is unable to determine whether the errors in the network are successfully corrected.


In this paper, we propose a Hermitian code based minimum storage regeneration (H-MSR) code and a Hermitian code based minimum bandwidth regeneration (H-MBR) code.
The proposed H-MSR/H-MBR codes can correct more errors than the RS-MSR/RS-MBR codes and can always determine whether the error correction is successful.  Our design is based on the structural analysis of the Hermitian code and the efficient decoding algorithm proposed in~\cite{Hermitian}.

It is worthwhile to point out that although there are strong connections between regenerating code in distributed storage and general network communication of which security problems have been well studied, 
our proposed H-MSR/H-MBR codes are fundamentally different from these security studies of network communication 
e.g.~\cite{sbjaggi0,sbjaggi1,sbjaggi2,sbjaggi3}, for two main reasons. 
First, the significant error correction capability of the proposed H-MSR/H-MBR codes is due to the underlying Hermitian code~\cite{Hermitian}, instead of relying on an error-detection layer, and/or shared secret keys between the sender and the receiver for error detection~\cite[Section 8.6.1]{sbjaggi0}. 
Second, the regenerating codes studied in this paper and the general network communication are different in that besides the overall data reconstruction, the regenerating codes also emphasize repairing of the corrupted code components (regeneration), while general network communication only focuses on data reproducing (reconstruction).  
Therefore, both the principle and the scope of this paper are different from the researches of security in general network communication.

\section{Preliminary and Assumptions} \label{Sec:Preliminary}

\subsection{Regenerating Code}
Regenerating code introduced in~\cite{Dimakis} is a linear code over the finite field $\mathbb{F}_q$ with a set of parameters $\{n, k, d, \alpha, \beta, B\}$. A file of size $B$ is stored in $n$ storage nodes, each of which stores $\alpha$ symbols. A replacement node can regenerate the contents of a failed node by downloading $\beta$ symbols from each of $d$ randomly selected storage nodes. So the total bandwidth needed to regenerate a failed node is $\gamma = d\beta$. The data collector (DC) can reconstruct the whole file by downloading $\alpha$ symbols from each of $k\leq d$ randomly selected storage nodes. In~\cite{Dimakis}, the following theoretical bound was derived:
\begin{equation}
\label{eq:min_cut}
B \leq \sum_{i=0}^{k-1}\min \{ \alpha, (d-i)\beta \}.
\end{equation}
From equation~(\ref{eq:min_cut}), a tradeoff between the regeneration bandwidth $\gamma$ and the storage requirement $\alpha$ was derived.
There are two special cases: minimum storage regeneration (MSR) point in which the storage parameter $\alpha$ is minimized;
\begin{equation}
\label{eq:MSR_tradeoff}
(\alpha_{MSR},\gamma_{MSR})= \left(\frac Bk, \frac{Bd}{k(d-k+1)}\right),
\end{equation}
and minimum bandwidth regeneration (MBR) point in which the bandwidth $\gamma$ is minimized:
\begin{equation}
\label{eq:MBR_tradeoff}
(\alpha_{MBR},\gamma_{MBR})= \left(\frac{2Bd}{2kd-k^2 + k},\frac{2Bd}{2kd-k^2 + k} \right).
\end{equation}

In this paper, we assume that DC keeps the encoding matrix secret and each storage node only knows its own encoding vector.

\subsection{Hermitian Code}\label{sec:Hermitian_Code}

A Hermitian curve $\calH(q)$ over $\mathbb{F}_{q^2}$ in affine coordinates is defined by:
\begin{equation}
\label{eq:hermitian_curve}
\calH(q):y^q + y = x^{q+1}.
\end{equation}
The genus of $\calH(q)$ is $\varrho =(q^2-q)/2$ and there are
$q^3$ points that satisfy equation~(\ref{eq:hermitian_curve}),  denoted as $P_{0,0},\cdots,P_{0,q-1},\cdots$, $P_{q^2-1,0},\cdots,P_{q^2-1,q-1}$ (See Table~\ref{tb:points}), where $\theta_0,\theta_1,\cdots,\theta_{q-1}$ are the $q$ solutions to $y^q + y = 0$ and $\phi$ is a primitive element in $\mathbb{F}_{q^2}$.  $L(mQ)$ is defined as:
\begin{eqnarray}
L(mQ) &=& \{ f_0(x) + yf_1(x) + \cdots + y^{q-1}f_{q-1}(x) | \nonumber \\
& & \deg f_j(x) <\kappa(j), j = 0,1,\cdots,q-1 \},\ \ \ \ \ \
\end{eqnarray}
where 
\begin{equation}
\kappa(j) =  \max \{ t | tq + j(q+1) \leq m \} + 1,
\end{equation}
for $m \geq q^2 - 1$.
A codeword of the Hermitian code \cite{Hermitian} $\calH_m$ is defined as
\begin{equation}\label{Eq:H-Encode}
(\varrho (P_{0,0}),\cdots,\varrho (P_{0,q-1}), \cdots,\varrho (P_{q^2-1,0}), \cdots, \varrho (P_{q^2-1,q-1})),
\end{equation}
where $\varrho  \in L(mQ)$. The dimension of the message before encoding can be calculated as $\dim(\calH_m) = \sum_{j=0}^{j=q-1} (\deg f_j(x) + 1)$.


\begin{table*}[ht]
\centering
\caption{$q^3$ rational points of the Hermitian curve}
\label{tb:points}
\begin{tabular}{cccc}
\hline
$P_{0,0} = (0,\theta_0)$ & $P_{1,0} = (1,\phi + \theta_0)$ &$\cdots$ & $P_{q^2-1,0} = (\phi^{q^2-2},\phi^{(q^2-2)(q+1)+1)} + \theta_0$\\
$P_{0,1} = (0,\theta_1)$ & $P_{1,1} = (1,\phi + \theta_1)$ &$\cdots$ & $P_{q^2-1,1} = (\phi^{q^2-2},\phi^{(q^2-2)(q+1)+1)} + \theta_1$\\
$\vdots$ & $\vdots$ &$\ddots$ &$\vdots$ \\
$P_{0,q-1} = (0,\theta_{q-1})$ & $P_{1,q-1} = (1,\phi + \theta_{q-1})$ &$\cdots$ & $P_{q^2-1,q-1} = (\phi^{q^2-2},\phi^{(q^2-2)(q+1)+1)} + \theta_{q-1}$\\  \hline
\end{tabular}
\end{table*}

\subsection{Adversarial Model}

In this paper, we assume some network nodes may be corrupted due to hardware failure or communication errors, and/or be compromised by malicious users.  As a result, upon request, these nodes may send out incorrect responses to disrupt the data regeneration and reconstruction. 
The adversary model is the same as~\cite{Rashmi-err}, 
We assume that the malicious users can take full control of $\tau$ ($\tau \leq n$ and corresponds to $s$ in~\cite{Rashmi-err})  storage nodes and collude to perform attacks. 

We will refer these symbols as \emph{bogus} symbols without making distinction between the corrupted symbols and compromised symbols.
We will also use corrupted nodes, malicious nodes and compromised nodes interchangeably without making any distinction.

\section{Encoding H-MSR Code}\label{Sec:H-MSR}

In this section, we will analyze the
H-MSR code based on the MSR point with $d = 2k - 2=2\alpha$. The code based on the MSR point with $d > 2k - 2$ can be derived the same way through truncating operations.


Let
$\alpha_0, \cdots , \alpha_{q-1}$ be a strictly decreasing integer sequence satisfying $0 < \alpha_i \leq\kappa(i), \, 0 \leq i \leq q-1$, where $\alpha_i$ is the parameter $\alpha$ for the underlying regenerating code.
The least common multiple of $\alpha_0, \cdots, \alpha_{q-1}$ is $A$.
Suppose the data contains $B = A \sum_{i=0}^{q-1}(\alpha_i + 1)$ message symbols from the finite field $\mathbb{F}_{q^2}$. In practice, if the size of the actual data is larger than $B$ symbols, we can fragment it into blocks of size $B$ and process each block individually.

We arrange the $B$ symbols into two matrices $S,T$ as below:
\begin{equation}
\label{eqn:matrices_s_t}
S = \begin{bmatrix}
S_0\\
S_1\\
\vdots \\
S_{q-1}
\end{bmatrix},\ \ \
T = \begin{bmatrix}
T_0\\
T_1\\
\vdots \\
T_{q-1}
\end{bmatrix},
\end{equation}
where
\begin{eqnarray}
S_i &=&
[S_{i,1},S_{i,2}, \cdots, S_{i,A/\alpha_i}],
\nonumber \\
T_i &=&
[T_{i,1}, T_{i,2}, \cdots, T_{i,A/\alpha_i}].
\end{eqnarray}
$S_{i,j}, \, 0 \leq i \leq q-1, 1 \leq j \leq A/\alpha_i$, is a symmetric matrix of size $\alpha_i \times \alpha_i$ with the upper-triangular entries filled by data symbols. Thus $S_{i,j}$ contains $\alpha_i(\alpha_i + 1)/2$ symbols. The number of columns of each submatrix $S_i,\, 0 \leq i \leq q-1$, is the same: $\alpha_i \cdot A / \alpha_i = A$. The size of matrix $S$ is $(\sum_{i=0}^{q-1}\alpha_i) \times A$. So it contains $\sum_{i=0}^{q-1}(\alpha_i (\alpha_i + 1)/2)  A / \alpha_i =  (A  \sum_{i=0}^{q-1}(\alpha_i + 1))/2$ data symbols.

$T_{i,j}, \, 0 \leq i \leq q-1, 1 \leq j \leq A/\alpha_i$, is constructed the same as $S_{i,j}$.  So $T$ contains the other $(A\sum_{i=0}^{q-1}(\alpha_i + 1))/2$ data symbols.

\begin{definition}
For a Hermitian code $\calH_m$ over $\mathbb{F}_{q^2}$, we encode matrix $M_{\dim(\calH_m)\times A}=[M_1,M_2\cdots,M_{A}]$ by encoding each column $M_i,\ i=1,2,\cdots,A$, individually using $\calH_m$. The codeword matrix is defined as
\begin{equation}\label{Eq:CodewordMatrix}
\calH_m(M)=[\calH_m(M_1),\calH_m(M_2),\cdots,\calH_m(M_{A})],
\end{equation}
where $\calH_m(M_i)$ has the following form ($\varrho  \in L(mQ)$):
\begin{equation}\label{Eq:EncodeForm}
[\varrho (P_{0,0}),\cdots,\varrho (P_{0,q-1}),\cdots,\varrho (P_{q^2-1,0}),\cdots,\varrho (P_{q^2-1,q-1})]^T,
\end{equation}
and the elements of $M_i$ are viewed as the coefficients of the polynomials $f_0(x),\cdots,f_{q-1}(x)$ in $\varrho$ when $M_i$ is encoded.
\end{definition}

Let
\begin{equation}\label{eq:phi}
\Phi_i =
\begin{bmatrix}
1 & 0 & 0 & \cdots  & 0 \\
1 & 1 & 1 & \cdots  & 1 \\
1 & \phi & \phi^2 & \cdots  & \phi^{\alpha_i-1} \\
\vdots & \vdots & \vdots & \ddots  & \vdots \\
1 & \phi^{q^2-2}& (\phi^{q^2-2})^2 & \cdots  & (\phi^{q^2-2})^{\alpha_i - 1} \\
\end{bmatrix}
\end{equation}
be a Vandermonde matrix, where $\phi$ is the primitive element in $\mathbb{F}_{q^2}$ mentioned in section~\ref{sec:Hermitian_Code} and $0 \leq i \leq q-1$.

Define
\begin{equation}
\Delta = \begin{bmatrix}
\lambda_0 & 0 & \cdots &  0 \\
0       & \lambda_1 & \cdots & 0 \\
 \vdots & \vdots & \ddots & \vdots \\
0       & 0 & \cdots & \lambda_{q^2-1}
\end{bmatrix}
\end{equation}
to be a diagonal matrix comprised of $q^2$ elements, where $\lambda_i, \,0 \leq i \leq q^2-1$, is chosen using the following two criteria: (i) $\lambda_i \neq \lambda_j,\, \forall i \neq j,  \,0 \leq i,j \leq q^2-1$. (ii) Any $d_i=2\alpha_i$ rows of the matrix $[\Phi_i, \Delta \cdot \Phi_i]$, $0 \leq i \leq q-1$, are linearly independent.

We also define
\begin{equation}
\Lambda_i = \lambda_i I
\end{equation}
to be a $q \times q$ diagonal matrix for $0 \leq i \leq q^2-1$, where I is the $q \times q$ identical matrix.
And
\begin{equation}
\Gamma = \begin{bmatrix}
\Lambda_0 & 0 & \cdots &  0 \\
0       & \Lambda_1 & \cdots & 0 \\
 \vdots & \vdots & \ddots & \vdots \\
0       & 0 & \cdots & \Lambda_{q^2-1}
\end{bmatrix}
\end{equation}
is a $q^3\times q^3$ diagonal matrix formed by $q^2$ diagonal submatrices $\Lambda_0, \cdots, \Lambda_{q^2 -1}$.




For distributed storage, we encode each pair of matrices $(S,T)$ using Algorithm~\ref{alg:enc}.  We will name this encoding scheme as \emph{Hermitian-MSR code encoding}, or \emph{H-MSR code encoding}.

\begin{algorithms}Encoding H-MSR Code
\label{alg:enc}

\normalfont

\begin{namelist}{\textbf{Step n:}}
\item[\step{1}] Encode the data matrices $S,T$ defined in equation (\ref{eqn:matrices_s_t}) using a Hermitian code $\calH_m$ over $\mathbb{F}_{q^2}$ with parameters $\kappa(j)\ (0 \leq j \leq q-1)$ and $m\ (m \geq q^2 -1)$.  Denote the generated $q^3 \times A$ codeword matrices as $\calH_m(S)$ and $\calH_m(T)$.

\item[\step{2}]  Compute the $q^3 \times A$ codeword matrix $Y=\calH_m(S) + \Gamma\calH_m(T)$.

\item[\step{3}]  Divide $Y$ into $q^2$ submatrices $Y_0,\cdots,Y_{q^2-1}$ of size $q\times A$ and store each submatrix in a storage node as shown in Fig.~\ref{fig:store_codeword}.
\end{namelist}
\end{algorithms}

\begin{figure}
\centering
\includegraphics[width=1.5in]{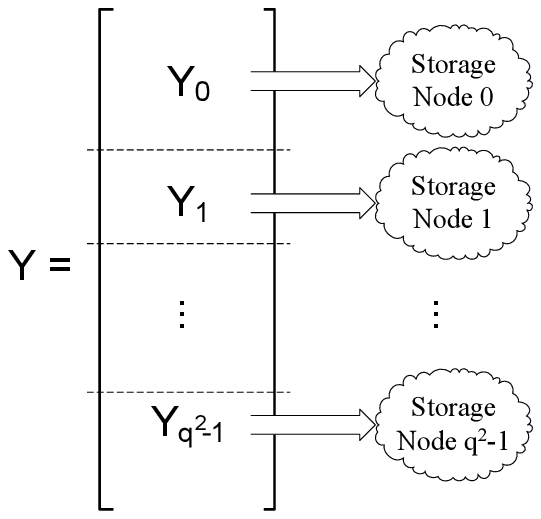}
\caption{Illustration of storing the codeword matrices in distributed storage nodes}
\label{fig:store_codeword}
\end{figure}

For H-MSR coding encoding, we have the following theorem.
\begin{theorem}
\label{th:msr}
The H-MSR code encoding described in Algorithm~\ref{alg:enc} can achieve the MSR point in distributed storage.
\end{theorem}

\begin{proof}
We first study the structure of the codeword matrix $\calH_m(S)$.
Since every column of the matrix is an independent Hermitian codeword, we can decode the first column $\mathbf{{h}} = [h_{0,0},\cdots,h_{0,q-1},\cdots,h_{q^2-1,0},$ $\cdots,h_{q^2-1,q-1}]^T$ as an example without loss of generality. We arrange the $q^3$ rational points of the Hermitian curve following the order in Table~\ref{tb:points}. In the table, we can find that for each $i,\, i=0,1,\cdots,q^2-1$, the rational points $P_{i,0}, P_{i,1},\cdots,P_{i,q-1}$ all have the same first coordinate.


Suppose $\varrho  \in L(mQ)$: $\varrho (P_{i,l}) = f_0(P_{i,l}) + y(P_{i,l})f_1(P_{i,l}) + \cdots + (y(P_{i,l}))^{q-1}f_{q-1}(P_{i,l}), \, 0 \leq i \leq q^2-1$, $0 \leq l \leq q-1, \deg f_j(x) = \alpha_j - 1$ for $0 \leq j \leq q-1$. Since $P_{i,0},P_{i,1},\cdots,P_{i,q-1}$ all have the same first coordinate and $f_j(P_{i,l})$ is only applied to the first coordinate of $P_{i,l}$, we have $f_j(P_{i,l})=f_j(\phi^{s_i}), \, s_0=-\infty, s_i=i-1, \mbox{for } i \geq 1, \phi^{-\infty}=0$, which does not depend on $l$.
Therefore, we can derive $q^2$ sets of equations for $0 \leq i \leq q^2 - 1$:
\begin{equation}
\label{equ_derive_f}
\left \{\begin{matrix}
f_0(\phi^{s_i}) + y(P_{i,0})f_1(\phi^{s_i}) + \cdots + (y(P_{i,0}))^{q-1}f_{q-1}(\phi^{s_i}) = h_{i,0}\\
f_0(\phi^{s_i}) + y(P_{i,1})f_1(\phi^{s_i}) + \cdots + (y(P_{i,1}))^{q-1}f_{q-1}(\phi^{s_i}) = h_{i,1}\\
\dotfill \\
f_0(\phi^{s_i}  )  +  y(P_{i,q-1} )f_1(\phi^{s_i} )  +  \cdots  +  (y(P_{i,q-1}))^{q-1}  f_{q-1}  (\phi^{s_i}  )  =  h_{i,q-1}
\end{matrix}.  \right.
\end{equation}\normalsize
If we store the codeword matrix in storage nodes according to Fig.~\ref{fig:store_codeword}, each set of the equations corresponds to a storage node. As we mentioned above, the set of equations in equation~(\ref{equ_derive_f}) can be derived in storage node $i$.

Since the coefficient matrix $B_i$ is a Vandermonde matrix:
\begin{equation}
B_i =
\begin{bmatrix}
1 & y(P_{i,0}) & \cdots  & y(P_{i,0})^{q-1} \\
1 & y(P_{i,1}) & \cdots  & y(P_{i,1})^{q-1} \\
\vdots & \vdots & \ddots & \vdots \\
1 & y(P_{i,q-1}) & \cdots  & y(P_{i,q-1})^{q-1} \\
\end{bmatrix}.
\end{equation}
we can solve $\mathbf{u}_i = [f_0(\phi^{s_i}), f_1(\phi^{s_i}), \cdots , f_{q-1}(\phi^{s_i})]^T$ from
\begin{equation}
\label{eqn_solve_f}
\mathbf{u}_i = B_i^{-1}  \mathbf{{h}}_i,
\end{equation}
where $\mathbf{{h}}_i = [h_{i,0}, h_{i,1}, \cdots , h_{i,q-1}]^T$.

From all the $q^2$ storage nodes, we can get vectors $\calF_i = [f_i(0),f_i(1),\cdots, f_i(\phi^{q^2-2})]^T,$ $i=0,\cdots,q-1$, which can be viewed as extended Reed-Solomon codes.

Now consider all the columns of $\calH_m(S)$, we can get the following equation:
\begin{equation}
\Phi_i  S_{i,j} = F_{i,j},
\end{equation}
where $F_{i,j}=[\calF_i^{(1)},\cdots,\calF_i^{(\alpha_i)}]$, $0 \leq i \leq q-1$, $1 \leq j \leq A/\alpha_i$, and $\calF_i^{(l)}$ corresponds to the $l^{th}$ column of the submatrix $S_{i,j}$.

Next we will consider the structure of the codeword matrix $\calH_m(T)$. Because the encoding process for $\calH_m(T)$ is the same as that of $\calH_m(S)$, for $\Gamma \calH_m(T)$,
we can derive
\begin{equation}
\Delta  \Phi_i  T_{i,j} = \Delta  E_{i,j},
\end{equation}
where $\calE_i=[e_i(0),e_i(1),\cdots,e_i(\phi^{q^2-2})]^T$,
$E_{i,j}=[\calE_i^{(1)},\cdots,\calE_i^{(\alpha_i)}]$, $0 \leq i \leq q-1$, $1 \leq j \leq A/\alpha_i$, and $\calE_i^{(l)}$ corresponds to the $l^{th}$ column of the submatrix $T_{i,j}$.

Thirdly, we will study the optimality of the code in the sense of the MSR point. For $\Phi_i S_{i,j} + \Delta \Phi_i T_{i,j},\, 0 \leq i \leq q-1,1 \leq j \leq A/\alpha_i$, since $S_{i,j},T_{i,j}$ are symmetric and satisfy the requirements for MSR point according to~\cite{Rashmi} with parameters $d= 2\alpha_i, k= \alpha_i + 1, \alpha = \alpha_i, \beta = 1, B= \alpha_i \cdot (\alpha_i + 1)$.
By encoding $S,T$ using $\calH_m(S) + \Gamma  \calH_m(T)$ and distributing $Y_0,\cdots,Y_{q^2-1}$ into $q^2$ storage nodes, each row of the matrix $\Phi_i  S_{i,j} + \Delta \Phi_i T_{i,j}, \, 0 \leq i \leq q-1,\, 1 \leq j \leq A/\alpha_i$,  can be derived in a corresponding storage node. Because $\Phi_i S_{i,j} + \Delta  \Phi_i T_{i,j}$ achieves the MSR point, data related to matrices $S_{i,j},T_{i,j}, \, 0 \leq i \leq q-1,\, 1 \leq j \leq A/\alpha_i$, can be regenerated at the MSR point. Therefore,  Algorithm~\ref{alg:enc} can achieve the MSR point.
\end{proof}

\section{Regeneration of the H-MSR Code}
\label{Sec:H-MSR-Regeneration}

In this section, we will first discuss regeneration of the H-MSR code in error-free network. Then we will discuss regeneration in hostile network.

\subsection{Regeneration in Error-free Network}
\label{Sec:H-MSR-Regeneration-error-free}

Let $\mathbf{v}_i=[e_0(\phi(^{s_i})), e_1(\phi(^{s_i})), \cdots, e_{q-1}(\phi(^{s_i}))]^T$, then
\begin{eqnarray}
\label{Eq:solve_y}
\mathbf{u}_i + \Lambda_i  \mathbf{v}_i =  B_i^{-1}  \mathbf{{y}}_i
= [f_0(\phi^{s_i}) +  \lambda_ie_0(\phi^{s_i}), \cdots , f_{q-1}(\phi^{s_i})+  \lambda_ie_{q-1}(\phi^{s_i})]^T,
\end{eqnarray}
for every column $\mathbf{y}_i$ of $Y_i$.

The main idea of the regeneration algorithms is to regenerate $f_l(\phi^{s_i}) +  \lambda_ie_l(\phi^{s_i})$, $0 \leq l \leq q-1$, by downloading help symbols from $d_l = 2\alpha_l$ nodes, where $d_l$ represents the regeneration parameter $d$ for $f_l(\phi^{s_i}) +  \lambda_ie_l(\phi^{s_i})$ in the H-MSR code regeneration.

Suppose node $z$ fails, we devise Algorithm~\ref{alg:reg_err_free} in the network to regenerate the exact H-MSR code symbols of node $z$ in a replacement node $z'$. For convenience, we suppose $d_q = 2\alpha_q = 0$ and define
\begin{equation} \label{Eq:V-Definition}
\mathbf{V}_{i,j,l}=\begin{bmatrix}
\nu_{i,l}\\
\nu_{i+1,l}\\
\vdots \\
\nu_{j,l}
\end{bmatrix},
\end{equation}
where $\mathbf{\nu}_{t,l}, \, i\leq t\leq j$, is the $t^{th}$ row of $[\Phi_l, \Delta  \Phi_l]$. Each node $i$, $0 \leq i \leq q^2 - 1$, only stores its own encoding vector $\nu_{i,l}$, $0 \leq l \leq q-1$.

First, replacement node $z'$ will send requests to helper nodes for regeneration: $z'$ sends the integer $j$ to $d_j-d_{j+1}$ helper nodes, to which $z'$ has not sent requests before, for every $j$ from $q-1$ to $0$ in descending order.

Upon receiving the request integer $j$, helper node $i$ will calculate and send the help symbols as follows:
node $i$ will first calculate $\widetilde{Y}_i =  B_i^{-1}  Y_i$ to remove the coefficient matrix $B_i$ from the codeword matrix.
Since
the $(l+1)^{th}$ row of $\widetilde{Y}_i$ corresponds to the symbols related to $f_l(\phi^{s_i}) +  \lambda_ie_l(\phi^{s_i})$, for $0 \leq l \leq j$, node $i$ will divide the $(l+1)^{th}$ row of $\widetilde{Y}_i$ into $A/\alpha_{l}$ row vectors of the size $1 \times \alpha_l$: $[{\bf\tilde{y}}_{i,l,1},{\bf\tilde{y}}_{i,l,2},\cdots,{\bf\tilde{y}}_{i,l,A/\alpha_{l}}]$.
Then for every $0 \leq l \leq j$ and $1 \leq t \leq A/\alpha_{l}$, node $i$ will calculate the help symbol $\tilde{p}_{i,l,t} = {\bf\tilde{y}}_{i,l,t}  {\bf\mu}_{z,l}^T$, where ${\bf\mu}_{z,l}$ is the $z^{th}$ row of the encoding matrix $\Phi_l$ defined in equation~(\ref{eq:phi}). At last, node $i$ will send out all the calculated symbols $\tilde{p}_{i,l,t}$. Here $j$ indicates that $z'$ is requesting symbols $\tilde{p}_{i,l,t}$, $0 \leq l \leq j$ and $1 \leq t \leq A/\alpha_{l}$, calculated by $[f_0(\phi^{s_i}) +  \lambda_ie_0(\phi^{s_i}) , \cdots, f_j(\phi^{s_i})+  \lambda_ie_j(\phi^{s_i})]^T$

Since $d_{l_1} > d_{l_2}$ for $l_1 < l_2$, for efficiency consideration, only $d_{q-1}$ helper nodes need to send out symbols $\tilde{p}_{i,l,t}$, $0 \leq l \leq q-1$ and $1 \leq t \leq A/\alpha_{l}$, calculated by $[f_0(\phi^{s_i}) +  \lambda_ie_0(\phi^{s_i}) , f_1(\phi^{s_i})+  \lambda_ie_1(\phi^{s_i}), \cdots , f_{q-1}(\phi^{s_i})+  \lambda_ie_{q-1}(\phi^{s_i})]^T$. Then $d_j - d_{j+1}$ nodes only need to send out symbols $\tilde{p}_{i,l,t}$, $0 \leq l \leq j$ and $1 \leq t \leq A/\alpha_{l}$, calculated by $[f_0(\phi^{s_i}) +  \lambda_ie_0(\phi^{s_i}) , f_1(\phi^{s_i})+  \lambda_ie_1(\phi^{s_i}), \cdots , f_j(\phi^{s_i})+  \lambda_ie_j(\phi^{s_i})]^T$ for $0 \leq j \leq q-2$. In this way, the total number of helper nodes that send out symbols $\tilde{p}_{i,l,t}$, $1 \leq t \leq A/\alpha_{l}$, calculated by $f_l(\phi^{s_i}) +  \lambda_ie_l(\phi^{s_i})$ is $d_{q-1} + \sum_{j=l}^{q-2}(d_j - d_{j+1}) = d_l$.

When the replacement node $z'$ receives all the requested symbols, it can regenerate the symbols stored in the failed node $z$ using the following algorithm:
\begin{algorithms}$z'$ regenerates symbols of the failed node $z$
\label{alg:reg_err_free}

\normalfont
\begin{namelist}{\textbf{Step n:}}
\item [\step{1}] For every $0 \leq l \leq q-1$ and $1 \leq t \leq A/\alpha_{l}$, calculate the regenerated symbols related to the help symbols $\tilde{p}_{i,l,t}$ from $d_l$ helper nodes. Without loss of generality, we assume $0 \leq i \leq d_l-1$:

\textbf{Step 1.1:} Let $\mathbf{{p}}=[\tilde{p}_{0,l,t}, \tilde{p}_{1,l,t}, \cdots, \tilde{p}_{d_l-1,l,t}]^T$, solve the equation: $\mathbf{V}_{0,d_l-1,l} \mathbf{x} =\mathbf{p}$.

\textbf{Step 1.2:} Since $\mathbf{x} = \begin{bmatrix}
S_{l,t}\\
T_{l,t}
\end{bmatrix}  \mu_{z,l}^T$ and $S_{l,t},T_{l,t}$ are symmetric, we can calculate $\mathbf{{x}}^T = [\mu_{z,l}  S_{l,t}, \mu_{z,l}  T_{l,t}]$.

\textbf{Step 1.3:} Compute $\mathbf{\tilde{y}}_{z,l,t}=\mu_{z,l} S_{l,t} + \lambda_z \mu_{z,l}T_{l,t} = \nu_{z,l}  \begin{bmatrix}
S_{l,t}\\
T_{l,t}
\end{bmatrix}$.


\item [\step{2}] Let $\widetilde{Y}_z$ be a $q \times A$ matrix with the $l^{th}$ row defined as $[\mathbf{\tilde{y}}_{z,l,1}, \cdots, \mathbf{\tilde{y}}_{z,l,A/\alpha_l}], 0 \leq l \leq q-1$.


\item [\step{3}] Calculate the regenerated symbols of the failed node $z$: $Y_{z'} = Y_z =  B_z  \widetilde{Y}_z$.
\end{namelist}
\end{algorithms}

From Algorithm~\ref{alg:reg_err_free}, we can derive the equivalent storage parameters for each symbol block of size $B_j = A  (\alpha_j +1)$: $d = 2\alpha_j$, $k = \alpha_j +1$, $\alpha = A$, $\beta = A/\alpha_j,\,  0 \leq j \leq q-1$ and equation~(\ref{eq:MSR_tradeoff}) of the MSR point holds for these parameters. Theorem \ref{th:msr} guarantees that Algorithm~\ref{alg:reg_err_free} can achieve the MSR point for data regeneration of the H-MSR code.

\subsection{Regeneration in Hostile Network}

In hostile network, Algorithm \ref{alg:reg_err_free} may not be able to regenerate the failed node due to possible bogus symbols received from the responses.
In fact,  even if the replacement node $z'$ can derive the symbol matrix $Y_{z'}$ using Algorithm~\ref{alg:reg_err_free}, it cannot verify the correctness of the result.

There are two modes for the helper nodes to regenerate the contents of a failed storage node in hostile network.
One mode is the detection mode, in which no error has been found in the symbols received from the helper nodes.
Once errors are detected, the recovery mode will be used to correct the errors and locate the malicious nodes.

\subsubsection{Detection Mode}

In the detection mode, the replacement node $z'$ will send requests in the way similar to that of the error-free network in Section~\ref{Sec:H-MSR-Regeneration-error-free}.
The only difference is that when $j=q-1$, $z'$ sends requests to $d_{q-1} - d_q + 1$ nodes instead of $d_{q-1} - d_q$ nodes. Helper nodes will still use the way similar to that of the error-free network in Section~\ref{Sec:H-MSR-Regeneration-error-free} to send the help symbols.
The regeneration algorithm is described in Algorithm~\ref{alg:reg_with_err_normal} with the detection probability characterized in Theorem~\ref{th:reg_with_err_normal}.

\begin{algorithms}[Detection mode] $z'$ regenerates symbols of the failed node $z$ in hostile network
\label{alg:reg_with_err_normal}

\normalfont
\begin{namelist}{\textbf{Step n:}}
\item [\step{1}] For every $0 \leq l \leq q-1$ and $1 \leq t \leq A/\alpha_{l}$, we can calculate the regenerated symbols which are related to the help symbols ${\tilde{p}_{i,l,t}}'$ from $d_l$ helper nodes. ${\tilde{p}_{i,l,t}}' = \tilde{p}_{i,l,t} + e_{i,l,t}$ is the response from the $i^{th}$ helper node.  If $\tilde{p}_{i,l,t}$ has been modified by the malicious node $i$, we have $e_{i,l,t} \in{\mathbb{F}_{q^2}}\backslash \{0\}$. Otherwise we have $e_{i,l,t} = 0$. To detect whether there are errors, we will calculate symbols from two sets of helper nodes then compare the results. (Without loss of generality, we assume $0 \leq i \leq d_l$.)

\textbf{Step 1.1:} Let ${\mathbf{{p}}_1}' = [{\tilde{p}_{0,l,t}}', {\tilde{p}_{1,l,t}}', \cdots, {\tilde{p}_{d_l-1,l,t}}']^T$, where the symbols are collected from node $0$ to node $d_l-1$, solve the equation $\mathbf{V}_{0,d_l-1,l} \mathbf{{x}}_1  = {\mathbf{{p}}_1}'$.

\textbf{Step 1.2:} Let ${\mathbf{{p}}_2}' = [{\tilde{p}_{1,l,t}}', {\tilde{p}_{2,l,t}}', \cdots, {\tilde{p}_{d_l,l,t}}']^T$, where the symbols are collected from node $1$ to node $d_l$, solve the equation $\mathbf{V}_{1,d_l,l} \mathbf{{x}}_2 = {\mathbf{{p}}_2}'$.

\textbf{Step 1.3:} Compare $\mathbf{{x}}_1$ with $\mathbf{{x}}_2$. If they are the same, compute $\mathbf{\tilde{y}}_{z,l,t}=\mu_{z,l}  S_{l,t} + \lambda_z  \mu_{z,l}  T_{l,t}$ as described in Algorithm~\ref{alg:reg_err_free}. Otherwise, errors are detected in the help symbols. Exit the algorithm and switch to recovery regeneration mode.

\item [\step{2}] No error has been detected for the calculating of the regeneration so far.
%
Let $\widetilde{Y}_z$ be a $q\times A$ matrix with the $l^{th}$ row defined as $[\mathbf{\tilde{y}}_{z,l,1}, \cdots, \mathbf{\tilde{y}}_{z,l,A/\alpha_l}], 0 \leq l \leq q-1$.
%


\item [\step{3}] Calculate the regenerated symbols of the failed node $z$: $Y_{z'} = Y_z =  B_z  \widetilde{Y}_z$.
\end{namelist}
\end{algorithms}

\begin{lemma}
\label{lm:reg_with_err_normal}
Suppose $e_0,\cdots,e_{d_l}$ are the $d_l + 1$ errors $e_{0,l,t},\cdots,e_{d_l,l,t}$ in Algorithm~\ref{alg:reg_with_err_normal}, $\hat{\mathbf{x}}_1 = \mathbf{V}_{0,d_l-1,l}^{-1}  \cdot [e_0,\cdots,e_{d_l-1} ]^T$ and $\hat{\mathbf{x}}_2 = \mathbf{V}_{1,d_l,l}^{-1}  \cdot [e_1,\cdots,e_{d_l} ]^T$. When the number of malicious nodes in the $d_l + 1$ helper nodes is less than $d_l + 1$, the probability that $\hat{\mathbf{x}}_1 = \hat{\mathbf{x}}_2 $ is at most $1/q^2$.
\end{lemma}

\begin{proof}
Since $\mathbf{V}_{0,d_l-1,l}$ and $\mathbf{V}_{1,d_l,l}$ are full rank matrices, we can get their corresponding inverse matrices. $\hat{\mathbf{x}}_1 = \hat{\mathbf{x}}_2$ is equivalent to $\mathbf{V}_{0,d_l-1,l} \cdot \hat{\mathbf{x}}_1 = \mathbf{V}_{0,d_l-1,l} \cdot \hat{\mathbf{x}}_2$.

First, we have
\begin{equation}
\mathbf{V}_{0,d_l-1,l} \cdot \hat{\mathbf{x}}_1 = [e_0, e_1, \cdots,e_{d_l-1} ]^T.
\end{equation}
Suppose $\mathbf{V}_{1,d_l,l}^{-1} =[\eta_0, \eta_1, \cdots, \eta_{d_l-1}]$, then we have:
\begin{equation}
\label{eq:psi_repair2_inv}
\nu_{r,l} \cdot  \eta_s = \left\{\begin{matrix}
 1,& r=s + 1\\
 0,& r \neq s +1
\end{matrix}\right., \, 1 \leq r\leq d_l,0 \leq s\leq d_l-1.
\end{equation}
\begin{eqnarray}
\mathbf{V}_{0,d_l-1,l}\cdot \hat{\mathbf{x}}_2
&=& \mathbf{V}_{0,d_l-1,l}\cdot \mathbf{V}_{1,d_l,l}^{-1} \cdot [e_1, e_2, \cdots,e_{d_l} ]^T  \nonumber\\
&=& \mathbf{V}_{0,d_l-1,l} \cdot [\eta_0, \eta_1, \cdots, \eta_{d_l-1}] \cdot [e_1, e_2, \cdots,e_{d_l} ]^T\\
&=&[x_{2,0}, e_1,\cdots,e_{d_l-1}]^T.\nonumber
\end{eqnarray}

To calculate $x_{2,0}$, we first derive the expression of $\nu_{0,l}$. Because $\nu_{1,l}, \nu_{2,l}, \cdots ,\nu_{d_l,l} $ are linearly independent, they can be viewed as a set of bases of the $d_l$ dimensional linear space. So we have
\begin{equation}
\label{eqn:mali_solve_coe}
\nu_{0,l} = \sum_{r=1}^{r=d_l} \zeta_r \cdot \nu_{r,l},
\end{equation}
where $\zeta_r\ne 0, \, r=1,\cdots,d_l$, because any $d_l$ vectors out of $\nu_{0,l}, \nu_{1,l}, \cdots ,\nu_{d_l,l} $ are linearly independent. Then
\begin{eqnarray}
x_{2,0} &=&
\left(\sum_{r=1}^{r=d_l} \zeta_r \cdot \nu_{r,l}\right)[\eta_0, \eta_1, \cdots, \eta_{d_l-1}] [e_1, e_2, \cdots,e_{d_l} ]^T\nonumber \\
&=& \sum_{r=1}^{r=d_l} \zeta_r \cdot e_r.
\end{eqnarray}

If
\begin{equation}
\label{eq:msr_reg_detect}
e_0 =  \sum_{r=1}^{r=d_l} \zeta_r \cdot e_r,
\end{equation}
then $\mathbf{V}_{0,d_l-1,l} \cdot \hat{\mathbf{x}}_1=\mathbf{V}_{0,d_l-1,l} \cdot \hat{\mathbf{x}}_2$ and $\hat{\mathbf{x}}_1=\hat{\mathbf{x}}_2$.

The number of errors corresponds to the number of malicious nodes. 
When only one element of $e_0, e_1, \cdots,e_{d_l}$ is nonzero, since $\zeta_1, \cdots,\zeta_{d_l}$ are all nonzero, equation~(\ref{eq:msr_reg_detect}) will never hold. In this case, the probability is $0$. When there are more than one nonzero elements, it means there are more than one malicious nodes. If the number of malicious nodes is less than $d_l + 1$, they will not be able to collude to solve the coefficients $\zeta_r$ in~(\ref{eqn:mali_solve_coe}). The probability that equation~(\ref{eq:msr_reg_detect}) holds will be $1/q^2$.
\end{proof}

\begin{theorem}[H-MSR Regeneration--Detection Mode]
\label{th:reg_with_err_normal}
When the number of malicious nodes in the $d_l + 1$ helper nodes of Algorithm~\ref{alg:reg_with_err_normal} is less than $d_l + 1$, the probability for the bogus symbols sent from the malicious nodes to be detected is at least $1-1/q^2$.
\end{theorem}

\begin{proof}
Since $\mathbf{V}_{0,d_l-1,l}$ and $\mathbf{V}_{1,d_l,l}$ are full rank matrices, $\mathbf{x}_1$ can be calculated by (For convenience, use $e_i$ to represent $e_{i,l,t}$):
\begin{eqnarray}
\mathbf{{x}}_1  &=& \mathbf{V}_{0,d_l-1,l}^{-1}  \cdot
\begin{bmatrix}
\tilde{p}_{0,l,t} + e_0,
\cdots,
\tilde{p}_{d_l-1,l,t} + e_{d_l-1}
\end{bmatrix}^T
\nonumber \\
&=& \mathbf{{x}} +
\mathbf{V}_{0,d_l-1,l}^{-1}  \cdot [e_0, e_1, \cdots,e_{d_l-1} ]^T \nonumber \\
& =& \mathbf{x} +  \hat{\mathbf{x}}_1.
\end{eqnarray}

$\mathbf{x}_2$ can be calculated the same way:
\begin{eqnarray}
\mathbf{x}_2  &=& \mathbf{x} +
\mathbf{V}_{1,d_l,l}^{-1}  \cdot
[e_1, e_2, \cdots,e_{d_l} ]^T
  =\mathbf{x} +  \hat{\mathbf{x}}_2.
\end{eqnarray}
If $\hat{\mathbf{x}}_1 = \hat{\mathbf{x}}_2$, Algorithm~\ref{alg:reg_with_err_normal} will fail to detect the errors. So we will focus on the relationship between $\hat{\mathbf{x}}_1$ and $\hat{\mathbf{x}}_2$. 
According to Lemma~\ref{lm:reg_with_err_normal}, when the number of malicious nodes in the $d_l + 1$ helper nodes is less than $d_l + 1$, the probability that $\hat{\mathbf{x}}_1 = \hat{\mathbf{x}}_2$ is at most $1/q^2$. So the probability that $\mathbf{x}_1 \neq \mathbf{x}_2$, equivalently the detection probability, is at least $1- 1/q^2$.
\end{proof}


\subsubsection{Recovery Mode}

Once the replacement node $z'$ detects errors using Algorithm~\ref{alg:reg_with_err_normal}, it will send integer $j = q-1$ to all the other $q^2 - 1$ nodes in the network requesting help symbols. Helper node $i$ will send help symbols similar to Section~\ref{Sec:H-MSR-Regeneration-error-free}. $z'$ can regenerate symbols using Algorithm~\ref{alg:reg_with_err_recovery}.

\begin{algorithms}[Recovery mode] $z'$ regenerates symbols of the failed node $z$ in hostile network
\label{alg:reg_with_err_recovery}

\normalfont
\begin{namelist}{\textbf{Step n:}}
\item [\step{1}] For every $q-1 \geq l \geq 0$ in descending order and $1 \leq t \leq A/\alpha_{l}$ in ascending order, we can regenerate the symbols when the errors in the received help symbols ${\tilde{p}_{i,l,t}}'$ from $q^2 - 1$ helper nodes can be corrected. Without loss of generality, we assume $0 \leq i \leq q^2-2$.

\textbf{Step 1.1:} Let $\mathbf{{p}}'=[{\tilde{p}_{0,l,t}}', {\tilde{p}_{1,l,t}}', \cdots,  {\tilde{p}_{q^2-2,l,t}}']^T$. Since $\mathbf{V}_{0,q^2-2,l} \cdot \mathbf{{x}}  = \mathbf{p}'$,  $\mathbf{p}'$ can be viewed as an MDS code with parameters $(q^2-1,d_l,q^2-d_l)$.

\textbf{Step 1.2:} Substitute ${\tilde{p}_{i,l,t}}'$ in $\mathbf{{p}}'$ with the symbol $\otimes$ representing an erasure if node $i$ has been detected to be corrupted in the previous loops (previous values of $l,t$).

\textbf{Step 1.3:} If the number of erasures in $\mathbf{{p}}'$ is larger than $\min \{q^2- d_l - 1, \lfloor (q^2- d_{q-1} - 1)/2 \rfloor\}$, then the number of errors have exceeded the error correction capability.  So here we will flag the decoding failure and exit the algorithm.

\textbf{Step 1.4:} Since the number of errors is within the error correction capability of the MDS code, decode $\mathbf{{p}}'$ to $\mathbf{p}_{cw}'$ and solve $\mathbf{{x}}$.

\textbf{Step 1.5:} If the $i^{th}$ position symbols of $\mathbf{{p}}'_{cw}$ and $\mathbf{{p}}'$ are different, mark node $i$ as corrupted.

\textbf{Step 1.6:} Compute $\mathbf{\tilde{y}}_{z,l,t}=\mu_{z,l} \cdot S_{l,t} + \lambda_z \cdot \mu_{z,l} \cdot T_{l,t}$ as described in Algorithm \ref{alg:reg_err_free}.

%
%

\item [\step{2}] Let $\widetilde{Y}_z$ be a $q\times A$ matrix with the $l^{th}$ row defined as $[\mathbf{\tilde{y}}_{z,l,1}, \cdots, \mathbf{\tilde{y}}_{z,l,A/\alpha_l}], 0 \leq l \leq q-1$.
%


\item [\step{3}] Calculate the regenerated symbols of the failed node $z$: $Y_{z'} = Y_z =  B_z  \widetilde{Y}_z$.
\end{namelist}
\end{algorithms}

For data regeneration described in Algorithm~\ref{alg:reg_with_err_recovery}, we have the following theorem:
\begin{theorem}[H-MSR Regeneration--Recovery Mode]\label{thm:reg_with_err_recovery}
For data regeneration, the number of errors that the H-MSR code can correct is 
\begin{equation}\label{eq:num_of_errs_msr_recovery}
\tau_{H-MSR}  =  q \cdot \lfloor (q^2 - {d_{q-1}} -1)/2  \rfloor.
\end{equation}
\end{theorem}
\begin{proof}
H-MSR code can be viewed as $q$ MDS codes with parameters ($q^2-1, d_l,q^2-d_l$), $l=0,\cdots, q-1$. Since $\alpha_l \leq\kappa(l)$ and $\kappa(l)$ is strictly decreasing, we can choose the sequence $\alpha_l$ to be strictly decreasing. So $d_l$ is also strictly decreasing. For the $q$ MDS codes, the minimum distance of the ($q^2-1,d_{q-1},q^2-d_{q-1}$) code is the largest. In Algorithm~\ref{alg:reg_with_err_recovery}, this code is decoded first and it can correct up to $\tau_{q-1} = \left \lfloor (q^2-d_{q-1}-1)/2 \right \rfloor$ errors, where $\left \lfloor x \right \rfloor$ is the floor function of $x$. Then the code $(q^2-1,d_l,q^2-d_l),\, l=q-2, \cdots, 0$, will be decoded sequentially.
The ($q^2-1, d_l,q^2-d_l$) code can correct at most $\tau_l = \tau_{q-1}$ errors when $q^2- d_0 - 1 \geq \tau_{q-1}$.
Thus, the total numbers errors that the H-MSR code can correct is $\tau_{H-MSR} = q \cdot {\tau_{q-1}} = q \cdot \lfloor (q^2 - {d_{q-1}} -1)/2  \rfloor$.
\end{proof}

\section{Reconstruction of the H-MSR code}
\label{Sec:H-MSR-Reconstruct}
In this section, we will first discuss reconstruction of the H-MSR code in error-free network. Then we will discuss reconstruction of the H-MSR code when there are corrupted nodes in the network.

\subsection{Reconstruction in Error-free Network}
\label{Sec:H-MSR-Reconstruct-error-free}

The main idea of the reconstruction algorithms is to reconstruct $f_l(\phi^{s_i}) +  \lambda_ie_l(\phi^{s_i})$, $0 \leq l \leq q-1$, by downloading help symbols from $k_l = \alpha_l + 1$ nodes, where $k_l$ is used to represent the reconstruction parameter $k$ for $f_l(\phi^{s_i}) +  \lambda_ie_l(\phi^{s_i})$ in the H-MSR code reconstruction.
We devise Algorithm~\ref{alg:recon_err_free} in the network for the data collector DC to reconstruct the original file. For convenience, we suppose $\alpha_q=0$.

First, DC will send requests to the storage nodes for reconstruction: DC sends integer $j$ to $k_j - k_{j+1}$ helper nodes, to which DC has not sent requests before, for every $j$ from $q-1$ to $0$ in descending order.

%

Upon receiving the request integer $j$, node $i$ will calculate and send symbols as follows: first node $i$ will calculate $\tilde{Y_i} =  B_i^{-1} Y_i$ to remove the coefficient matrix $B_i$ from the codeword matrix. Since
the $(l+1)^{th}$ row of $\tilde{Y_i}$ corresponds to the symbols related to $f_l(\phi^{s_i}) +  \lambda_ie_l(\phi^{s_i})$, for $0 \leq l \leq j$, node $i$ will send out the $(l+1)^{th}$ row of $\tilde{Y_i}$: $\mathbf{\tilde{y}}_{i,l}$. Here $j$ indicates that DC is requesting symbols of $\mathbf{\tilde{y}}_{i,l}$, $0 \leq l \leq j$, calculated by $[f_0(\phi^{s_i}) +  \lambda_ie_0(\phi^{s_i}) , \cdots , f_j(\phi^{s_i})+  \lambda_ie_j(\phi^{s_i})]^T$.

Since $k_{l_1} > k_{l_2}$ for $l_1 < l_2$, for efficiency consideration, only $k_{q-1}$ helper nodes need to send out symbols of $\mathbf{\tilde{y}}_{i,l}$, $0 \leq l \leq q-1$, calculated by $[f_0(\phi^{s_i}) +  \lambda_ie_0(\phi^{s_i}) , f_1(\phi^{s_i})+  \lambda_ie_1(\phi^{s_i}), \cdots , f_{q-1}(\phi^{s_i})+  \lambda_ie_{q-1}(\phi^{s_i})]^T$. Then $k_j - k_{j+1}$ nodes only need to send out symbols of $\mathbf{\tilde{y}}_{i,l}$, $0 \leq l \leq j$, calculated by $[f_0(\phi^{s_i}) +  \lambda_ie_0(\phi^{s_i}) , f_1(\phi^{s_i})+  \lambda_ie_1(\phi^{s_i}), \cdots , f_j(\phi^{s_i})+  \lambda_ie_j(\phi^{s_i})]^T$ for $0 \leq j \leq q-2$. In this way, the total number of helper nodes that send out symbols of $\mathbf{\tilde{y}}_{i,l}$ calculated by $f_l(\phi^{s_i}) +  \lambda_ie_l(\phi^{s_i})$ is $k_{q-1} + \sum_{j=l}^{q-2}(k_j - k_{j+1}) = k_l$.


When DC receives all the requested symbols, it can reconstruct the original file using the following algorithm:
\begin{algorithms}DC reconstructs the original file
\label{alg:recon_err_free}

\normalfont
\begin{namelist}{\textbf{Step n:}}
\item [\step 1] For every  $0 \leq l \leq q-1$ , divide the response symbol vector $\mathbf{\tilde{y}}_{i,l}$ from the $i^{th}$ node into $A/\alpha_{l}$ equal row vectors: $[\mathbf{\tilde{y}}_{i,l,1},\mathbf{\tilde{y}}_{i,l,2},\cdots,\mathbf{\tilde{y}}_{i,l,A/\alpha_{l}}]$, $0 \leq i \leq k_l-1$.

\item [\step 2] For every $0 \leq l \leq q-1$ and $1 \leq t \leq A/\alpha_{l}$, DC reconstructs the matrices related to the original file:

\textbf{Step 2.1:} Let $R=[\mathbf{\tilde{y}}_{0,l,t}^T, \mathbf{\tilde{y}}_{1,l,t}^T, \cdots, \mathbf{\tilde{y}}_{k_l-1,l,t}^T]^T$, we have the equation: $\mathbf{V}_{0,k_l-1,l} \cdot \begin{bmatrix}
S_{l,t}\\
T_{l,t}
\end{bmatrix}  =  R$
according to the encoding algorithm.

\textbf{Step 2.2:} DC reconstructs $S_{l,t},T_{l,t}$ using techniques similar to \cite{Rashmi}.

\item [\step 3] DC reconstructs the original file from all the matrices $S_{l,t},T_{l,t}$, $0 \leq l \leq q-1$ and $1 \leq t \leq A/\alpha_{l}$.
\end{namelist}
\end{algorithms}

%
%
%
%
%

\subsection{Reconstruction in Hostile Network}

Similar to the regeneration algorithms, the reconstruction algorithms in error-free network do not work in hostile network. Even if the data collecter can calculate the symbol matrices $S,T$ using Algorithm~\ref{alg:recon_err_free}, it cannot verify whether the result is correct or not. There are two modes for the original file to be reconstructed in hostile network. One mode is the detection mode, in which no error has been found in the symbols received from the storage nodes. Once errors are detected in the detection mode, the recovery mode will be used to correct the errors and locate the malicious nodes.

\subsubsection{Detection Mode}

In the detection mode, DC will send requests in the way similar to that for the error-free network in Section~\ref{Sec:H-MSR-Reconstruct-error-free}. The only difference is that when $j=q-1$, DC will send requests to $k_{q-1} - k_q + 1$ nodes instead of $k_{q-1} - k_q$ nodes. Storage nodes will still use the way similar to that for the error-free network in Section~\ref{Sec:H-MSR-Reconstruct-error-free} to send symbols. The reconstruction algorithm is described in Algorithm~\ref{alg:recon_with_err_normal} with the detection probability described in Theorem~\ref{th:recon_with_err_normal}.

\begin{algorithms}[Detection mode] DC reconstructs the original file in hostile network
\label{alg:recon_with_err_normal}

\normalfont
\begin{namelist}{\textbf{Step n:}}
\item [\step 1] For every  $0 \leq l \leq q-1$ , we can divide the symbol vector ${\mathbf{\tilde{y}}_{i,l}}'$ into $A/\alpha_{l}$ equal row vectors: $[{\mathbf{\tilde{y}}_{i,l,1}}',{\mathbf{\tilde{y}}_{i,l,2}}',\cdots,{\mathbf{\tilde{y}}_{i,l,A/\alpha_{l}}}']$. ${\mathbf{\tilde{y}}_{i,l}}'= \mathbf{\tilde{y}}_{i,l} + \mathbf{ e}_{i,l}$ is the response from the $i^{th}$ storage node. If $\mathbf{\tilde{y}}_{i,l}$ has been modified by the malicious node $i$, we have $\mathbf{e}_{i,l} \in (\mathbb{F}_{q^2})^{A} \backslash \{\mathbf{0}\}$. To detect whether there are errors, we will reconstruct the original file from two sets of storage nodes then compare the results. (Without loss of generality, we assume $0 \leq i \leq k_l$.)

\item [\step 2]For every $0 \leq l \leq q-1$ and $1 \leq t \leq A/\alpha_{l}$, DC can reconstruct the matrices related to the original file:

\textbf{Step 2.1:} Let
 ${R}'=[{\mathbf{\tilde{y}}_{0,l,t}}'^T, {\mathbf{\tilde{y}}_{1,l,t}}'^T, \cdots, {\mathbf{\tilde{y}}_{k_l,l,t}}'^T]^T$.

\textbf{Step 2.2:} Let
${R_1}' = [{\mathbf{\tilde{y}}_{0,l,t}}'^T, {\mathbf{\tilde{y}}_{1,l,t}}'^T, \cdots, {\mathbf{\tilde{y}}_{\alpha_l,l,t}}'^T]^T$, which are the symbols collected from node $0$ to node $k_l-1=\alpha_l$, then we have $\mathbf{V}_{0,\alpha_l,l}\cdot \begin{bmatrix}
S_1\\
T_1
\end{bmatrix} ={R_1}'$. Solve $S_1,T_1$ using the method same to algorithm~\ref{alg:recon_err_free}.

\textbf{Step 2.3:} Let
${R_2}' = [{\mathbf{\tilde{y}}_{0,l,t}}'^T, \cdots, {\mathbf{\tilde{y}}_{\alpha_l - 1,l,t}}'^T, {\mathbf{\tilde{y}}_{\alpha_l + 1,l,t}}'^T]^T$, which are the symbols collected from node $0$ to node $k_l=\alpha_l + 1$ except node $\alpha_l$,
and
$\Psi_{DC2}=\begin{bmatrix}
\nu_{0,l}\\
\vdots \\
\nu_{\alpha_l - 1,l}\\
\nu_{\alpha_l + 1,l}
\end{bmatrix}$, then we have $\Psi_{DC2} \cdot \begin{bmatrix}
S_2\\
T_2
\end{bmatrix}= {R_2}'$. Solve $S_2,T_2$ using the method same to algorithm~\ref{alg:recon_err_free}.

\textbf{Step 2.4:} Compare $[S_1,T_1]$ with $[S_2,T_2]$. If they are the same, let $[S_{l,t},T_{l,t}] = [S_1,T_1]$. Otherwise, errors are detected in the received symbols. Exit the algorithm and switch to recovery reconstruction mode.

\item [\step 3] No error has been detected for the calculating of the reconstruction so far.
So DC can reconstruct the original file from all the matrices $S_{l,t},T_{l,t}$, $0 \leq l \leq q-1$ and $1 \leq t \leq A/\alpha_{l}$.
\end{namelist}
\end{algorithms}

\begin{theorem}[H-MSR Reconstruction--Detection Mode]
\label{th:recon_with_err_normal}

When the number of malicious nodes in the $k_l + 1$ nodes of Algorithm~\ref{alg:recon_with_err_normal} is less than $k_l + 1$, the probability for the bogus symbols sent from the malicious nodes to be detected is at least $1-(1/q^2)^{2(\alpha_l - 2)}$.
\end{theorem}

\begin{proof}
We arrange this proof as follows. We will first study the requirements for $S_1 = S_2,T_1 = T_2$ in Algorithm~\ref{alg:recon_with_err_normal} which will lead to the failure of the Algorithm when there are bogus symbols. Then we will study the corresponding failure probabilities depending on different values of $\lambda_i$ of the matrix $\Delta$ defined in section~\ref{Sec:H-MSR}.

For convenience we write $\mathbf{ e}_{i,l,t}$ as $\mathbf{ e}_i$ in the proof. $\mathbf{ e}_i \in {[\mathbb{F}_{q^2}]^{\alpha_l}}$ for $0 \leq i \leq \alpha_l + 1$. We also write $\Psi_{DC} = [\Phi_{DC}, \Delta_{DC} \cdot \Phi_{DC}]$, where $\Phi_{DC} = \begin{bmatrix}
\mu_{0}\\
\mu_{1}\\
\vdots \\
\mu_{k_l-1}
\end{bmatrix}$ and $\mu_{i}$ represents ${\mu}_{i,l}$ which is the $i^{th}$ row of the encoding matrix $\Phi_l$ defined in section~\ref{Sec:H-MSR}.$\\$

\paragraph*{Step 1. Derive the requirements}

For ${R_1}'=R_1 + W_1$ in Algorithm~\ref{alg:recon_with_err_normal}, we have:
\begin{equation}
\label{eq:recon_raw}
\Phi_{DC1}  S_1  \Phi_{DC1}^T + \Delta_{DC1}  \Phi_{DC1}  T_1  \Phi_{DC1}^T =   R_1\Phi_{DC1}^T + W_1\Phi_{DC1}^T,
\end{equation}
where $\Phi_{DC1} = \begin{bmatrix}
\mu_{0}\\
\mu_{1}\\
\vdots \\
\mu_{\alpha_l}
\end{bmatrix}$, $W_1 = \begin{bmatrix}
\mathbf{ e}_{0}\\
\mathbf{ e}_{1}\\
\vdots \\
\mathbf{ e}_{\alpha_l}
\end{bmatrix}$. Suppose $C_1=\Phi_{DC1}  S_1  \Phi_{DC1}^T, D_1 = \Phi_{DC1}  T_1  \Phi_{DC1}^T$, we can write equation~(\ref{eq:recon_raw}) as:
\begin{equation}
C_1 + \Delta_{DC1}  D_1 =   R_1\Phi_{DC1}^T + W_1\Phi_{DC1}^T  = \hat{R}_1 + \hat{W}_1.
\end{equation}
It is easy to see that $C_1$ and $D_1$ are symmetric, so we have
\begin{equation}
\label{eq:recon}
\left\{\begin{matrix}
C_{1,i,j} + \lambda_i \cdot D_{1,i,j}= \hat{R}_{1,i,j} + \hat{W}_{1,i,j}\\
C_{1,i,j} + \lambda_j \cdot D_{1,i,j}= \hat{R}_{1,j,i} + \hat{W}_{1,j,i}
\end{matrix}\right.,
\end{equation}
where $C_{1,i,j}, D_{1,i,j}, \hat{R}_{1,i,j}, \hat{W}_{1,i,j}$ are the elements in the $i^{th}$ row, $j^{th}$ column of $C_1, D_1, \hat{R}_1, \hat{W}_1$ respectively. Solve equation~(\ref{eq:recon}) for all the $i,j$ ($i \neq j, 0 \leq i \leq \alpha_l, 0 \leq j \leq \alpha_l-1$), we can get the corresponding $C_{1,i,j}, D_{1,i,j}$. Because the structure of $C_1$ and $D_1$ are the same, we will only focus on $C_1$ (corresponding to $S_1$) in the proof. The calculation for $D_1$ (corresponding to $T_1$) is the same.
\begin{equation}
\Phi_{DC1}  S_1  \Phi_{DC1}^T = \begin{bmatrix}
\mu_{0}\\
\mu_{1}\\
\vdots \\
\mu_{\alpha_l}
\end{bmatrix} \cdot S_1 \cdot
[\mu_{0}^T, \mu_{1}^T, \cdots , \mu_{\alpha_l}^T] = C_1.
\end{equation}
So the elements of the $i^{th}$ row of $C_1$ (except the element in the diagonal position) can be written as:
\begin{equation}
\begin{matrix}
\mu_i \cdot S_1 \cdot [\mu_{0}^T,\cdots,  \mu_{i-1}^T,  \mu_{i+1}^T \cdots , \mu_{\alpha_l}^T]
=
[C_{1,i,0}, \cdots, C_{1,i,i-1}, C_{1,i,i+1}, \cdots, C_{1,i,\alpha_l}].
\end{matrix}
\end{equation}

Let $\Omega=\begin{bmatrix}
\mu_0\\ \mu_1\\ \vdots \\ \mu_{\alpha_l-1}
\end{bmatrix}$, then $\Omega$ is an $\alpha_l \times \alpha_l$ full rank matrix, and we can derive $S_1$ from
\begin{equation}
\Omega \cdot S_1  =
\begin{bmatrix}
[C_{1,0,1}, C_{1,0,2},\cdots, C_{1,0,\alpha_l}][\mu_1^T, \mu_2^T,\cdots, \mu_{\alpha_l}^T]^{-1}\\
[C_{1,1,0}, C_{1,1,2},\cdots, C_{1,1,\alpha_l}][\mu_0^T, \mu_2^T,\cdots, \mu_{\alpha_l}^T]^{-1}\\
\cdots \\
[C_{1,\alpha_l-1,0}, C_{1,\alpha_l-1,1},\cdots, C_{1,\alpha_l-1,\alpha_l}][\mu_0^T, \mu_1^T,\cdots, \mu_{\alpha_l}^T]^{-1}\\
\end{bmatrix}.
\end{equation}

For ${R_2}'=R_2 + W_2$ in Algorithm~\ref{alg:recon_with_err_normal}, we can get $\Omega \cdot S_2$ the same way. If $\Omega\cdot S_1 = \Omega\cdot S_2$, Algorithm~\ref{alg:recon_with_err_normal} will fail to detect the errors. This will happen if all the rows of $\Omega\cdot S_1$ and $\Omega\cdot S_2$ are the same. So we will focus on the $i^{th}$ row of $\Omega\cdot S_1$ and $\Omega\cdot S_2$.

\paragraph*{Step 2. Calculate the failure probabilities}

Depending on the values of $\lambda_i$, we discuss two cases:

(a) If none of the $\lambda_i$ ($0 \leq i \leq \alpha_l$) equals to 0, we can solve $C_{1,i,j}$ in equation~(\ref{eq:recon}):

\begin{eqnarray}
\label{eq:solve_C}
C_{1,i,j} &=& \frac{\lambda_j \cdot \hat{R}_{1,i,j} - \lambda_i \cdot \hat{R}_{1,j,i} }{\lambda_i \cdot \lambda_j} + \frac{\mathbf{ e}_i \cdot \mu_j^T}{\lambda_i} - \frac{\mathbf{ e}_j \cdot \mu_i^T}{\lambda_j} \nonumber \\ &=& N_{1,i,j} + Q_{1,i,j}.
\end{eqnarray}
In equation~(\ref{eq:solve_C}), $N_{1,i,j}$ represents the original solution without errors, while $Q_{1,i,j}$ represents the impact of the errors. So the $i^{th}$ row of $\Omega\cdot S_1$ can be written as:
\begin{eqnarray}
&&[C_{1,i,0},\cdots,C_{1,i,i-1},C_{1,i,i+1},\cdots,C_{1,i,\alpha_l}]\cdot\Pi_{1,i}^{-1} \nonumber \\
&=& [N_{1,i,0},\cdots,N_{1,i,i-1},N_{1,i,i+1},\cdots,N_{1,i,\alpha_l}]\cdot\Pi_{1,i}^{-1} \nonumber \\
&&+ [Q_{1,i,0},\cdots,Q_{1,i,i-1},Q_{1,i,i+1},\cdots,Q_{1,i,\alpha_l}]\cdot\Pi_{1,i}^{-1}  \label{eq:solve_C_nom_plus_err} \\
&=& \xi_i + \delta_{1,i},\nonumber
\end{eqnarray}
where $\Pi_{1,i} = [\mu_0^T,\cdots,\mu_{i-1}^T, \mu_{i+1}^T,\cdots, \mu_{\alpha_l}^T]$. $\xi_i$ corresponds to the part independent of the errors. $\delta_{1,i}$ is the error part and can be further expanded as:
\begin{eqnarray}
\delta_{1,i} &=& \left[\frac{\mathbf{e}_i \cdot \mu_0^T}{\lambda_i},\cdots,\frac{\mathbf{ e}_i \cdot \mu_{i-1}^T}{\lambda_i}, \frac{\mathbf{ e}_i \cdot \mu_{i+1}^T}{\lambda_i},\cdots,\frac{\mathbf{ e}_i \cdot \mu_{\alpha_l}^T}{\lambda_i}\right] \cdot \Pi_{1,i}^{-1} \nonumber\\
&&- \left[\frac{\mathbf{ e}_0 \cdot \mu_i^T}{\lambda_0},\cdots, \frac{\mathbf{ e}_{i-1} \cdot \mu_i^T}{\lambda_{i-1}},\frac{\mathbf{ e}_{i+1} \cdot \mu_i^T}{\lambda_{i+1}},\cdots, \frac{\mathbf{ e}_{\alpha_l} \cdot \mu_i^T}{\lambda_{\alpha_l}}\right] \cdot \Pi_{1,i}^{-1}.
\label{eq:solve_C_err}
\end{eqnarray}
The first part of equation~(\ref{eq:solve_C_err}) can be reduced as follows:
\begin{eqnarray}
&&\left[\frac{\mathbf{ e}_i \cdot \mu_0^T}{\lambda_i},\cdots,\frac{\mathbf{ e}_i \cdot \mu_{i-1}^T}{\lambda_i}, \frac{\mathbf{ e}_i \cdot \mu_{i+1}^T}{\lambda_i},\cdots,\frac{\mathbf{ e}_i \cdot \mu_{\alpha_l}^T}{\lambda_i}\right] \cdot \Pi_{1,i}^{-1} \nonumber\\
&=&\frac{\mathbf{ e}_i}{\lambda_i} \cdot \left[\mu_0^T,\cdots,\mu_{i-1}^T, \mu_{i+1}^T,\cdots, \mu_{\alpha_l}^T\right] \cdot \Pi_{1,i}^{-1} \label{eq:solve_C_err_1st_part}\\
&=& \frac{\mathbf{ e}_i}{\lambda_i}. \nonumber
\end{eqnarray}
So we have:
\begin{eqnarray}
\delta_{1,i} &=& \frac{\mathbf{ e}_i}{\lambda_i} -  \left[\frac{\mathbf{ e}_0 \cdot \mu_i^T}{\lambda_0},\cdots, \frac{\mathbf{ e}_{i-1} \cdot  \mu_i^T}{\lambda_{i-1}},\frac{\mathbf{ e}_{i+1} \cdot\mu_i^T}{\lambda_{i+1}},\cdots, \frac{\mathbf{ e}_{\alpha_l} \cdot \mu_i^T}{\lambda_{\alpha_l}}\right] \cdot \Pi_{1,i}^{-1}
\nonumber \\
&=&\frac{\mathbf{ e}_i}{\lambda_i} - \rho_{1,i}.
\end{eqnarray}

For ${R_2}'=R_2 + W_2$ in Algorithm~\ref{alg:recon_with_err_normal} where $W_2 = \begin{bmatrix}
\mathbf{ e}_{0}\\
\vdots \\
\mathbf{ e}_{\alpha_l-1}\\
\mathbf{ e}_{\alpha_l+1}
\end{bmatrix}$, we can derive $C_{2,i,j}$, then $\Omega\cdot S_2$ the same way. The $i^{th}$ row of $\Omega\cdot S_2$ can be written as:
\begin{equation}
\xi_i + \delta_{2,i} = \xi_i + \frac{\mathbf{ e}_i}{\lambda_i} - \rho_{2,i},
\end{equation}
where $\rho_{2,i} = \left[\frac{\mathbf{ e}_0 \cdot \mu_i^T}{\lambda_0},\cdots, \frac{\mathbf{ e}_{i-1}\cdot \mu_i^T}{\lambda_{i-1}},\frac{\mathbf{ e}_{i+1} \cdot\mu_i^T}{\lambda_{i+1}},\cdots, \frac{\mathbf{e}_{\alpha_l-1} \cdot \mu_i^T}{\lambda_{\alpha_l-1}},\frac{\mathbf{ e}_{\alpha_l+1} \cdot \mu_i^T}{\lambda_{\alpha_l+1}}\right]\cdot\Pi_{2,i}^{-1}$, $\Pi_{2,i} = [\mu_0^T,\cdots,\mu_{i-1}^T, \mu_{i+1}^T$, $\cdots, \mu_{\alpha_l-1}^T, \mu_{\alpha_l+1}^T]$.

Because $\Pi_{1,i}$ is a full rank matrix, $\rho_{1,i} = \rho_{2,i}$ is equivalent to $\rho_{1,i} \cdot \Pi_{1,i} = \rho_{2,i} \cdot \Pi_{1,i}$. Similar to the proof of Lemma~\ref{lm:reg_with_err_normal}, suppose $\Pi_{2,i}^{-1} = \begin{bmatrix}
\eta_{0}\\
\vdots \\
\eta_{\alpha_l-1}\\
\eta_{\alpha_l+1}
\end{bmatrix}$, we have $\eta_s \cdot \mu_r^T = \left\{\begin{matrix}
 1 \:\:r=s\\
 0 \:\:r\neq s
\end{matrix}\right.$. So
\begin{eqnarray}
\rho_{1,i}  \cdot  \Pi_{1,i}  &=& \left[\cdots, \frac{\mathbf{ e}_{i-1} \cdot  \mu_i^T}{\lambda_{i-1}},\frac{\mathbf{ e}_{i+1}  \cdot \mu_i^T}{\lambda_{i+1}},\cdots, \frac{\mathbf{ e}_{\alpha_l-1}  \cdot  \mu_i^T}{\lambda_{\alpha_l-1}},\frac{\mathbf{ e}_{\alpha_l}  \cdot  \mu_i^T}{\lambda_{\alpha_l}}\right],\\
\rho_{2,i}  \cdot  \Pi_{1,i}  &=& \left[\cdots, \frac{\mathbf{ e}_{i-1} \cdot  \mu_i^T}{\lambda_{i-1}},\frac{\mathbf{ e}_{i+1}  \cdot \mu_i^T}{\lambda_{i+1}},\cdots, \frac{\mathbf{ e}_{\alpha_l-1}  \cdot  \mu_i^T}{\lambda_{\alpha_l-1}}, x_{2,\alpha_l}\right].
\end{eqnarray}

Because $\mu_0^T,\cdots,\mu_{i-1}^T, \mu_{i+1}^T,\cdots, \mu_{\alpha_l-1}^T, \mu_{\alpha_l+1}^T$ are linearly independent, they can be viewed as a set of bases of the $\alpha_l$ dimensional linear space. So we have
\begin{equation}
\label{eqn:mali_col_recon}
\mu_{\alpha_l}^T = \sum_{r=0,r \neq i,\alpha_l}^{r=\alpha_l + 1} \zeta_r \cdot \mu_r^T.
\end{equation}
Thus
\begin{eqnarray}
x_{2,\alpha_l} &=& \left[\cdots, \frac{\mathbf{e}_{i-1}\cdot  \mu_i^T}{\lambda_{i-1}},\frac{\mathbf{ e}_{i+1} \cdot\mu_i^T}{\lambda_{i+1}},\cdots, \frac{\mathbf{ e}_{\alpha_l-1}\cdot \mu_i^T}{\lambda_{\alpha_l-1}},\frac{\mathbf{ e}_{\alpha_l+1}  \cdot  \mu_i^T}{\lambda_{\alpha_l+1}}\right] \cdot
\Pi_{2,i}^{-1}  \cdot  \left(\sum_{r=0,r \neq i,\alpha_l}^{r=\alpha_l + 1} \zeta_r  \cdot  \mu_r^T\right) \nonumber \\
&=& \left(\sum_{r=0,r \neq i,\alpha_l}^{r=\alpha_l + 1} \zeta_r  \cdot \frac{\mathbf{e}_r  \cdot  \mu_i^T}{\lambda_r}\right).
\end{eqnarray}
If
\begin{equation}
\label{eq:msr_recon_detect}
\frac{\mathbf{ e}_{\alpha_l}  \cdot  \mu_i^T}{\lambda_{\alpha_l}} =  \sum_{r=0,r \neq i,\alpha_l}^{r=\alpha_l + 1}  \zeta_r  \cdot  \frac{\mathbf{ e}_r  \cdot  \mu_i^T}{\lambda_r} \:\:(0 \leq i \leq \alpha_l-1),
\end{equation}
$\rho_{1,i}$ and $\rho_{2,i}$ will be equal, so are $\Omega\cdot S_1$ and $\Omega\cdot S_2$. Therefore, Algorithm~\ref{alg:recon_with_err_normal} will fail.

For the error $\mathbf{ e}_i\ (0 \leq i \leq \alpha_l + 1)$, the following equation holds:
\begin{equation}
\mathbf{e}_i \cdot [\mu_0^T, \mu_1^T,\cdots,\mu_{\alpha_l - 1}^T] = [\hat{e}_{i,0}, \hat{e}_{i,1}, \cdots , \hat{e}_{i,\alpha_l - 1}] = \hat{\mathbf{e}}_i.
\end{equation}
Because $[\mu_0^T, \mu_1^T,\cdots,\mu_{\alpha_l - 1}^T]$ is a full rank matrix, there is a one-to-one mapping between $\mathbf{e}_i$ and $\hat{\mathbf{e}}_i$. 
 Equation~(\ref{eq:msr_recon_detect}) can be written as:
\begin{equation}
\label{eq:msr_recon_detect_number}
\frac{\hat{e}_{\alpha_l,i}}{\lambda_{\alpha_l}} =  \sum_{r=0,r \neq i,\alpha_l}^{r=\alpha_l + 1}  \zeta_r  \cdot  \frac{\hat{e}_{r,i}}{\lambda_r} \:\:(0 \leq i \leq \alpha_l-1).
\end{equation}
When the number of malicious nodes in the $k_l +1$ nodes is less than $k_l +1$, the malicious nodes can collude to satisfy equation~(\ref{eq:msr_recon_detect_number}) for at most one particular $i$. So the probability that equation~(\ref{eq:msr_recon_detect_number}) holds is $1/q^2$ for at least $\alpha_l - 1$ out of  $\alpha_l$ $i's$ between $0$ and $\alpha_l - 1$. If we consider equation~(\ref{eq:msr_recon_detect_number}) for all the $i's$ simultaneously, the probability will be at most $(1/q^2)^{\alpha_l - 1}$. As discussed before, the probability for $T_1 = T_2$ will be at most $(1/q^2)^{\alpha_l - 1}$. In this case, the detection probability is at least $1-(1/q^2)^{2(\alpha_l-1)}$.


(b) If one of the $\lambda_i$ ($0 \leq i \leq \alpha_l$) equals to $0$, we can assume $\lambda_0 = 0$ without loss of generality.
When $i=0$, the solution for equation~(\ref{eq:recon}) is:
\begin{equation}
C_{1,0,j} = \hat{R}_{1,0,j} + \mathbf{ e}_0 \cdot \mu_j^T = N_{1,0,j} + Q_{1,0,j}.
\end{equation}
Similar to equations~(\ref{eq:solve_C_nom_plus_err}), (\ref{eq:solve_C_err}) and (\ref{eq:solve_C_err_1st_part}), we have $\delta_{1,0} = \mathbf{ e}_0$. For ${R_2}'=R_2 + W_2$, it is easy to see that $\delta_{2,0} = \mathbf{ e}_0$. So the first rows of $\Omega\cdot S_1$ and $\Omega\cdot S_2$ are the same no matter what the error vector $\mathbf{e}_0$ is.

When $i>0,j=0$, the solution for equation~(\ref{eq:recon}) is:
\begin{equation}
C_{1,i,0} = \hat{R}_{1,i,0} + \mathbf{ 0} \cdot \mu_0^T + \mathbf{ e}_0 \cdot \mu_i^T = N_{1,i,0} + Q_{1,i,0},
\end{equation}
where $\mathbf{ 0}$ is a zero row vector. When $i>0,j>0$, the solution has the same expression as equation~(\ref{eq:solve_C}). In this case, for the $i^{th}\: (i>0)$ row of $\Omega\cdot S_1$, equation~(\ref{eq:solve_C_err}) can be written as:
\begin{eqnarray}
\delta_{1,i} &=& \left[\mathbf{ 0},\cdots,\frac{\mathbf{ e}_i \cdot \mu_{i-1}^T}{\lambda_i}, \frac{\mathbf{ e}_i \cdot \mu_{i+1}^T}{\lambda_i},\cdots,\frac{\mathbf{ e}_i \cdot \mu_{\alpha_l}^T}{\lambda_i}\right] \cdot \Pi_{1,i}^{-1} \nonumber\\
&&-\left[-\mathbf{e}_0 \cdot \mu_i^T,\cdots, \frac{\mathbf{ e}_{i-1} \cdot \mu_i^T}{\lambda_{i-1}},\frac{\mathbf{ e}_{i+1} \cdot \mu_i^T}{\lambda_{i+1}},\cdots, \frac{\mathbf{ e}_{\alpha_l} \cdot \mu_i^T}{\lambda_{\alpha_l}}\right] \cdot \Pi_{1,i}^{-1}.
\label{eq:solve_C_err_case_b}
\end{eqnarray}
The first part of equation~(\ref{eq:solve_C_err_case_b}) can be divided into two parts:
\begin{eqnarray}
&& \left[\frac{\mathbf{e}_i \cdot \mu_0^T}{\lambda_i},\cdots,\frac{\mathbf{e}_i \cdot \mu_{i-1}^T}{\lambda_i}, \frac{\mathbf{e}_i \cdot \mu_{i+1}^T}{\lambda_i},\cdots,\frac{\mathbf{e}_i \cdot \mu_{\alpha_l}^T}{\lambda_i}\right] \cdot \Pi_{1,i}^{-1}
-  \left[\frac{\mathbf{e}_i \cdot \mu_0^T}{\lambda_i},\mathbf{0},\cdots,\mathbf{0} \right]  \cdot  \Pi_{1,i}^{-1}\nonumber \\
& =&
\frac{\mathbf{ e}_i}{\lambda_i}  -  \frac{\mathbf{ e}_i}{\lambda_i} \cdot  [\mu_0^T, \mathbf{ 0},\cdots,\mathbf{ 0}]\cdot \Pi_{1,i}^{-1}.
\end{eqnarray}
So equation~(\ref{eq:solve_C_err_case_b}) can be further written as:
\begin{eqnarray}
\delta_{1,i} &=& \frac{\mathbf{e}_i}{\lambda_i} -\left[\frac{\mathbf{ e}_i \cdot \mu_0^T}{\lambda_i}-\mathbf{ e}_0 \cdot \mu_i^T,\cdots, \frac{\mathbf{ e}_{i-1} \cdot \mu_i^T}{\lambda_{i-1}},\frac{\mathbf{ e}_{i+1} \cdot \mu_i^T}{\lambda_{i+1}},\cdots,
\frac{\mathbf{ e}_{\alpha_l} \cdot \mu_i^T}{\lambda_{\alpha_l}}\right] \cdot \Pi_{1,i}^{-1} \nonumber \\
&=& \frac{\mathbf{ e}_i}{\lambda_i} - \rho_{1,i}.
\end{eqnarray}

By employing the same derivation in case~(a), for $1 \leq i \leq \alpha_l - 1$, $\rho_{1,i}$ and $\rho_{2,i}$ will be equal if
\begin{eqnarray}
\label{eq:msr_recon_detect_case_b}
\frac{\mathbf{ e}_{\alpha_l}  \cdot  \mu_i^T}{\lambda_{\alpha_l}} &=&  \sum_{r=1,r \neq i,\alpha_l}^{r=\alpha_l + 1}  \zeta_r  \cdot  \frac{\mathbf{ e}_r  \cdot  \mu_i^T}{\lambda_r} - \zeta_0  \cdot  \mathbf{ e}_0  \cdot  \mu_i^T + \zeta_0  \cdot  \frac{\mathbf{ e}_i  \cdot  \mu_0^T}{\lambda_i},\\
\label{eq:msr_recon_detect_number_case_b}
\frac{\hat{e}_{\alpha_l,i}}{\lambda_{\alpha_l}} &=&  \sum_{r=1,r \neq i,\alpha_l}^{r=\alpha_l + 1}  \zeta_r  \cdot  \frac{\hat{e}_{r,i}}{\lambda_r} - \zeta_0  \cdot  \hat{e}_{0,i} + \zeta_0  \cdot  \frac{\hat{e}_{i,0}}{\lambda_i}.
\end{eqnarray}
When the number of malicious nodes in the $k_l + 1$ nodes is less than $k_l + 1$, similar to case~(a), the probability that equation~(\ref{eq:msr_recon_detect_number_case_b}) holds is $1/q^2$ for at least $\alpha_l - 2$ out of  $\alpha_l - 1$ $i's$ between $1$ and $\alpha_l - 1$. If we consider equation~(\ref{eq:msr_recon_detect_number_case_b}) for all the $i's$ simultaneously, the probability will be at most $(1/q^2)^{\alpha_l - 2}$. Here the probability for $T_1 = T_2$ will be at most $(1/q^2)^{\alpha_l-2}$. In this case, the detection probability is $1-(1/q^2)^{2(\alpha_l - 2)}$.


Combining both cases, the detection probability is at least $1-(1/q^2)^{2(\alpha_l - 2)}$.
\end{proof}


\subsubsection{Recovery Mode}

Once DC detects errors using Algorithm~\ref{alg:recon_with_err_normal}, it will send integer $j = q-1$ to all the $q^2$ nodes in the network requesting symbols. Storage nodes will still use the way similar to that of the error-free network in Section~\ref{Sec:H-MSR-Reconstruct-error-free} to send symbols. The reconstruct procedures are described in Algorithm~\ref{alg:recon_with_err_recovery}.

\begin{algorithms}[Recovery mode] DC reconstructs the original file in hostile network
\label{alg:recon_with_err_recovery}

\normalfont
\begin{namelist}{\textbf{Step n:}}
\item [\step 1] For every  $0 \leq l \leq q-1$ , we divide the symbol vector ${\mathbf{\tilde{y}}_{i,l}}'$ into $A/\alpha_{l}$ equal row vectors: $[{\mathbf{\tilde{y}}_{i,l,1}}',{\mathbf{\tilde{y}}_{i,l,2}}',\cdots,{\mathbf{\tilde{y}}_{i,l,A/\alpha_{l}}}']$. (Without loss of generality, we assume $0 \leq i \leq q^2-1$.)

\item [\step 2] For every $q-1 \geq l \geq 0$ in descending order and $1 \leq t \leq A/\alpha_{l}$ in ascending order, DC can reconstruct the matrices related to the original file when the errors in the received symbol vectors ${\mathbf{\tilde{y}}_{i,l,t}}'$ from $q^2$ storage nodes can be corrected:

\textbf{Step 2.1:} Let
${R}'=[{\mathbf{\tilde{y}}_{0,l,t}}'^T, {\mathbf{\tilde{y}}_{1,l,t}}'^T, \cdots, {\mathbf{\tilde{y}}_{q^2-1,l,t}}'^T]^T$.

\textbf{Step 2.2:} If the number of corrupted nodes detected is larger than $\min \{ q^2- k_l , \lfloor (q^2 - k_{q-1})/2 \rfloor\}$, then the number of errors have exceeded the error correction capability.  We will flag the decoding failure and exit the algorithm.

\textbf{Step 2.3:}  Since the number of errors is within the error correction capability of the H-MSR code, substitute ${\mathbf{\tilde{y}}_{i,l,t}}'$ in ${R}'$ with the symbol $\otimes$
representing an erasure vector if node $i$ has been detected to be corrupted in the previous loops (previous values of $l,t$).

\textbf{Step 2.4:} Solve $S_{l,t},T_{l,t}$ using the method described in section \ref{sec:Rec-ST}. If symbols from node $i$ are detected to be erroneous during the calculation, mark node $i$ as corrupted.

\item [\step 3] DC reconstructs the original file from all the matrices $S_{l,t},T_{l,t}$, $0 \leq l \leq q-1$ and $1 \leq t \leq A/\alpha_{l}$.
\end{namelist}
\end{algorithms}

For data reconstruction described in Algorithm~\ref{alg:recon_with_err_recovery}, we have the following theorem:
\begin{theorem}[H-MSR Reconstruction--Recovery Mode]\label{thm:recon_with_err_recovery}
For data reconstruction, the number of errors that the H-MSR code can correct is 
\begin{equation}\label{eq:num_of_errs_msr_recovery_recon}
\tau_{H-MSR}  =  q \cdot \lfloor (q^2 - {k_{q-1}})/2  \rfloor.
\end{equation}
\end{theorem}
\begin{proof}
Similar to the proof of Theorem~\ref{thm:reg_with_err_recovery} in data regeneration, for data reconstruction Algorithm~\ref{alg:recon_with_err_recovery}, H-MSR code can be viewed as $q$ MDS codes with parameters $(q^2-1, k_l - 1,q^2 - k_l + 1)$.  The decoding for the reconstruction is performed from the code with the largest minimum distance to the code with the smallest minimum distance as in the data regeneration case. So here we have similar result as in equation~(\ref{eq:num_of_errs_msr_recovery}).
\end{proof}

\subsection{Recover Matrices $S_{l,t}, T_{l,t}$ from $q^2$ Storage Nodes}\label{sec:Rec-ST}

When there are bogus symbols ${\tilde{p}_{i,l,t}}'$ sent by the corrupted nodes for certain $l,t$, we can recover the matrices $S_{l,t}, T_{l,t}$ as follows:

For ${R}'$ in Algorithm~\ref{alg:recon_with_err_recovery}, we have $\Psi_{DC} \cdot \begin{bmatrix}
{S}'\\
{T}'
\end{bmatrix}  =  {R}'$, and
\begin{equation}
\Phi_{DC}  {S}'  \Phi_{DC}^T + \Delta_{DC}  \Phi_{DC}  {T}'  \Phi_{DC}^T =   {R}'\Phi_{DC}^T,
\end{equation}
where $\Psi_{DC} = [\Phi_{DC}, \Delta_{DC} \cdot \Phi_{DC}]$, $\Phi_{DC} = \begin{bmatrix}
\mu_{0}\\
\mu_{1}\\
\vdots \\
\mu_{q^2-1}
\end{bmatrix}$ and $\mu_{i}$ represents ${\mu}_{i,l}$ which is the $i^{th}$ row of the encoding matrix $\Phi_l$ in the proof of Theorem~\ref{th:msr}.

Let $C=\Phi_{DC}  {S}'  \Phi_{DC}^T$, $D=\Phi_{DC}  {T}'  \Phi_{DC}^T$, and ${\hat{R}}'={R}'\Phi_{DC}^T$,  then
\begin{equation}
C + \Delta_{DC}  D  =   {\hat{R}}'.
\end{equation}
Since $C,D$ are both symmetric, we can solve the non-diagonal elements of them as follows:
\begin{equation}
\label{eq:recon_recovery}
\left\{\begin{matrix}
C_{i,j} + \lambda_i \cdot D_{i,j}= {\hat{R}}'_{i,j}\\
C_{i,j} + \lambda_j \cdot D_{i,j}= {\hat{R}}'_{j,i}
\end{matrix}\right..
\end{equation}
Because matrices $C$ and $D$ have the same structure, here we only focus on $C$ (corresponding to ${S}'$). It is straightforward to see that if node $i$ is malicious and there are errors in the $i^{th}$ row of ${R}'$, there will be errors in the $i^{th}$ row of ${\hat{R}}'$. Furthermore, there will be errors in the $i^{th}$ row and $i^{th}$ column of $C$. Define ${S}'\Phi_{DC}^T={\hat{S}}'$, we have
\begin{equation}
\Phi_{DC} {\hat{S}}' = C.
\end{equation}
Here we can view each column of $C$ as a $(q^2-1, \alpha_l, q^2 - \alpha_l )$ MDS code because $\Phi_{DC}$ is a Vandermonde matrix. The length of the code is $q^2 - 1$ since the diagonal elements of $C$ is unknown.
Suppose node $j$ is uncorrupted. If the number of erasures $\sigma$ (corresponding to the previously detected corrupted nodes) and the number of the corrupted nodes $\tau$ that have not been detected satisfy:
\begin{equation}
\sigma + 2\tau + 1 \leq q^2 - \alpha_l,
\end{equation}
then the $j^{th}$ column of $C$ can be recovered and the error locations (corresponding to the corrupted nodes) can be pinpointed. The non-diagonal elements of $C$ can be recovered. So DC can reconstruct $S_{l,t}$ using the method similar to \cite{Rashmi}. For $T_{l,t}$, the recovering process is similar.

\section{Encoding H-MBR Code}\label{Sec:H-MBR}

In this section, we will analyze the H-MBR code based on the MBR point with $\beta = 1$. According to equation~(\ref{eq:MBR_tradeoff}), we have $d=\alpha$.

Let $\alpha_0,\cdots,\alpha_{q-1}$ be a strictly decreasing integer sequence satisfying $0 < \alpha_i \leq \kappa(i),0 \leq i \leq q-1$. The least common multiple of $\alpha_0,\cdots,\alpha_{q-1}$ is $A$. Let $k_0,\cdots,k_{q-1}$ be a integer sequence satisfying $0 < k_i \leq \alpha_i,0 \leq i \leq q-1$. Suppose the data contains $B=A \cdot \sum_{i=0}^{q-1}{(k_i (2\alpha_i - k_i + 1)/(2 \alpha_i))}$ message symbols from the finite field $\mathbb{F}_{q^2}$. In practice, if the size of the actual data is larger than $B$ symbols, we can fragment it into blocks of size $B$ and process each block individually.

We arrange the $B$ symbols into matrix $M$ as below:
\begin{equation}
M = \begin{bmatrix}
M_0\\ M_1\\ \vdots\\ M_{q-1}
\end{bmatrix},
\end{equation}
where
\begin{equation}
M_i = [M_{i,1}, M_{i,2}, \cdots, M_{i,A/\alpha_i} ]
\end{equation}
and
\begin{equation}
\label{eq:mbr_m_matrix}
M_{i,j} = \begin{bmatrix}
S_{i,j} & T_{i,j} \\
T_{i,j}^T & \bf{0}.
\end{bmatrix}
\end{equation}
$S_{i,j}, 0 \leq i \leq q-1, 1 \leq j \leq A/\alpha_i$ is a symmetric matrix of size $k_i \times k_i$ with the upper-triangular entries filled by data symbols. $T_{i,j}$ is a $k_i \times (\alpha_i - k_i)$ matrix. Thus $M_{i,j}$ contains $k_i(2\alpha_i - k_i + 1)/2$ symbols, $M_i$ contains $A \cdot k_i(2\alpha_i - k_i + 1)/(2\alpha_i)$ symbols and $M$ contains $B$ symbols.

For distributed storage, we encode $M$ using Algorithm~\ref{alg:enc_MBR}:

\begin{algorithms}Encoding H-MBR Code
\label{alg:enc_MBR}

\normalfont
\begin{namelist}{\textbf{Step n:}}
\item [\step 1] First we encode the data matrices $M$ defined above using a Hermitian code $\calH_m$ over $\mathbb{F}_{q^2}$ with parameters $\kappa(j)\ (0 \leq j \leq q-1)$ and $m\ (m \geq q^2 -1)$. The $q^3 \times A$ codeword matrix can be written as $Y = \calH_m(M)$.

\item [\step 2] Then we divide the codeword matrix $Y$ into $q^2$ submatrices $Y_0,\cdots,Y_{q^2-1}$ of the size $q \times A$ and store one submatrix in each of the $q^2$ storage nodes as shown in Fig.~\ref{fig:store_codeword}.
\end{namelist}
\end{algorithms}

%

Then we have the following theorem:

\begin{theorem}
\label{th:mbr}
By processing the data symbols using Algorithm~\ref{alg:enc_MBR}, we can achieve the MBR point in distributed storage.
\end{theorem}

\begin{proof}
Similar to the proof of Theorem~\ref{alg:enc}, we can get the following equation considering all the columns of $\calH_m(M)$:
\begin{equation}
\Phi_i \cdot M_{i,j} = G_{i,j},
\end{equation}
where $G_{i,j}=[\calG_i^{(1)},\cdots,\calG_i^{(\alpha_i)}]$, $0 \leq i \leq q-1$, $1 \leq j \leq A/\alpha_i$. $\calG_i^{(l)}$ corresponds to the $l^{th}$ column of the submatrix $M_{i,j}$ and each element of $\calG_i=[g_i(0),g_i(1),\cdots,g_i(\phi^{q^2-2})]^T$ can be derived from a distinct storage node. $\Phi_i$ is defined in equation~(\ref{eq:phi}).

Next we will study the optimality of the code in the sense of the MBR point. For $\Phi_i \cdot M_{i,j},\, 0 \leq i \leq q-1,1 \leq j \leq A/\alpha_i$,
$M_{i,j}$ is symmetric and satisfies the requirements for MBR point according to~\cite{Rashmi} with parameters $d= \alpha_i, k= k_i, \alpha = \alpha_i, \beta = 1, B= k_i(2\alpha_i - k_i + 1)/2$.
By encoding $M$ using $\calH_m(M)$ and distributing $Y_0,\cdots,Y_{q^2-1}$ into $q^2$ storage nodes, each row of the matrix $\Phi_i \cdot M_{i,j}, \, 0 \leq i \leq q-1,\, 1 \leq j \leq A/\alpha_i$, can be derived in a corresponding storage node. Because $\Phi_i \cdot M_{i,j}$ achieves the MBR point, data related to matrices $M_{i,j}, \, 0 \leq i \leq q-1,\, 1 \leq j \leq A/\alpha_i$, can be regenerated at the MBR point. Therefore,  Algorithm~\ref{alg:enc_MBR} can achieve the MBR point.
\end{proof}

\section{Regeneration of the H-MBR Code}\label{Sec:H-MBR-Regeneration}
In this section, we will first discuss regeneration of the H-MBR code in error-free network. Then we will discuss regeneration in hostile network.

\subsection{Regeneration in Error-free Network}
\label{Sec:H-MBR-Regeneration-error-free}

Let $\mathbf{w}_i=[g_0(\phi(^{s_i})), g_1(\phi(^{s_i})), \cdots, g_{q-1}(\phi(^{s_i}))]^T$, then
$\mathbf{w}_i =  B_i^{-1} \cdot \mathbf{{y}}_i
= [g_0(\phi^{s_i}), \cdots , g_{q-1}(\phi^{s_i})]^T,$
for every column $\mathbf{y}_i$ of $Y_i$.

The main idea of the regeneration algorithms is similar to that of the H-MSR code: regenerate $g_l(\phi(^{s_i})),0\leq l \leq q-1$, by downloading help symbols from $d_l = \alpha_l$ nodes, where $d_l$ is the regeneration parameter $d$ for $g_l(\phi(^{s_i}))$ in the H-MBR code regeneration.

Suppose node $z$ fails, we use Algorithm~\ref{alg:reg_err_free_mbr} to regenerate the exact H-MBR code symbols of node $z$. For convenience, we suppose $d_q = \alpha_q = 0$ and define
\begin{equation} \label{Eq:W-Definition}
\mathbf{W}_{i,j,l}=\begin{bmatrix}
\mu_{i,l}\\
\mu_{i+1,l}\\
\vdots \\
\mu_{j,l}
\end{bmatrix},
\end{equation}
where $\mathbf{\mu}_{t,l}, \, i\leq t\leq j$, is the $t^{th}$ row of $\Phi_l$.

Similar to the H-MSR code, replacement node $z'$ will send requests to helper nodes in the way same to that in Section~\ref{Sec:H-MSR-Regeneration-error-free}. Upon receiving the request integer $j$, helper node $i$ will calculate and send the help symbols similar to that of Section~\ref{Sec:H-MSR-Regeneration-error-free}.

When the replacement node $z'$ receives all the requested symbols, it can regenerate the symbols stored in the failed node $z$ using the following algorithm:
\begin{algorithms}$z'$ regenerates symbols of the failed node $z$
\label{alg:reg_err_free_mbr}

\normalfont
\begin{namelist}{\textbf{Step n:}}
\item [\step 1] For every $0 \leq l \leq q-1$ and $1 \leq t \leq A/\alpha_{l}$, we can calculate the regenerated symbols which are related to the help symbols $\tilde{p}_{i,l,t}$ from $d_l$ helper nodes: (Without loss of generality, we assume $0 \leq i \leq d_l-1$.)

\textbf{Step 1.1:} Let $\mathbf{{p}}=[\tilde{p}_{0,l,t}, \tilde{p}_{1,l,t}, \cdots, \tilde{p}_{d_l-1,l,t}]^T$, solve the equation: $\mathbf{W}_{0,d_l-1,l}\cdot \mathbf{{x}} =\mathbf{{p}}$.

\textbf{Step 1.2:} Since $\mathbf{{x}} = M_{l,t} \cdot \mu_{z,l}^T$ and $M_{l,t}$ is symmetric, we can calculate $\mathbf{\tilde{y}}_{z,l,t} = \mathbf{{x}}^T = \mu_{z,l} \cdot M_{l,t}$.

\item [\step 2] Let $\widetilde{Y}_z$ be a $q\times A$ matrix with the $l^{th}$ row defined as $[\mathbf{\tilde{y}}_{z,l,1}, \cdots, \mathbf{\tilde{y}}_{z,l,A/\alpha_l}], 0 \leq l \leq q-1$.

\item [\step 3] Calculate the regenerated symbols of the failed node $z$: $Y_{z'} = Y_z =  B_z \cdot \widetilde{Y}_z$.
\end{namelist}
\end{algorithms}

%
%
%
%
%

For Algorithm~\ref{alg:reg_err_free_mbr} we can derive the equivalent storage parameters for each symbol block of size
$B_j = A k_j(2\alpha_j - k_j + 1)/(2\alpha_j): d = \alpha_j, k = k_j, \alpha = A, \beta = A/\alpha_j,\,  0 \leq j \leq q-1$ and equation~(\ref{eq:MBR_tradeoff}) of the MBR point holds for these parameters. Theorem \ref{th:mbr} guarantees that Algorithm~\ref{alg:reg_err_free_mbr} can achieve the MBR point for data regeneration of the H-MBR code.

\subsection{Regeneration in Hostile Network}
In hostile network, Algorithm \ref{alg:reg_err_free_mbr} may be unable to regenerate the failed node due to the possible bogus symbols received from the responses.
In fact,  even if the replacement node $z'$ can derive the symbol matrix $Y_{z'}$ using Algorithm~\ref{alg:reg_err_free_mbr}, it cannot verify the correctness of the result.

Similar to the H-MSR code, there are two modes for the helper nodes to regenerate the H-MBR code of a failed storage node in hostile network.
One mode is the detection mode, in which no error has been found in the symbols received from the helper nodes.
Once errors are detected, the recovery mode will be used to correct the errors and locate the malicious nodes.

\subsubsection{Detection Mode}

In the detection mode, the replacement node $z'$ will send requests in the way similar to that of the error-free network in Section~\ref{Sec:H-MBR-Regeneration-error-free}. The only difference is that when $j=q-1$, $z'$ sends requests to $d_{q-1} - d_q + 1$ nodes instead of $d_{q-1} - d_q$ nodes. Helper nodes will still use the way similar to that of the error-free network in Section~\ref{Sec:H-MBR-Regeneration-error-free} to send the help symbols. The regeneration algorithm is described in Algorithm~\ref{alg:reg_with_err_normal_mbr} with the detection probability characterized in Theorem~\ref{th:reg_with_err_normal_mbr}.

\begin{algorithms}[Detection mode] $z'$ regenerates symbols of the failed node $z$ in hostile network
\label{alg:reg_with_err_normal_mbr}

\normalfont
\begin{namelist}{\textbf{Step n:}}
\item [\step 1] For every $0 \leq l \leq q-1$ and $1 \leq t \leq A/\alpha_{l}$, we can calculate the regenerated symbols which are related to the help symbols ${\tilde{p}_{i,l,t}}'$ from $d_l$ helper nodes. ${\tilde{p}_{i,l,t}}' = \tilde{p}_{i,l,t} + e_{i,l,t}$ is the response from the $i^{th}$ helper node.  If $\tilde{p}_{i,l,t}$ has been modified by the malicious node $i$, we have $e_{i,l,t} \in{\mathbb{F}_{q^2}}\backslash \{0\}$. To detect whether there are errors, we will calculate symbols from two sets of helper nodes then compare the results. (Without loss of generality, we assume $0 \leq i \leq d_l$.)

\textbf{Step 1.1:} Let ${\mathbf{{p}}_1}' = [{\tilde{p}_{0,l,t}}', {\tilde{p}_{1,l,t}}', \cdots, {\tilde{p}_{d_l-1,l,t}}']^T$, where the symbols are collected from node $0$ to node $d_l-1$, solve the equation $\mathbf{W}_{0,d_l-1,l} \cdot \mathbf{{x}}_1  = {\mathbf{{p}}_1}'$.

\textbf{Step 1.2:} Let ${\mathbf{{p}}_2}' = [{\tilde{p}_{1,l,t}}', {\tilde{p}_{2,l,t}}', \cdots, {\tilde{p}_{d_l,l,t}}']^T$, where the symbols are collected from node $1$ to node $d_l$, solve the equation $\mathbf{W}_{1,d_l,l}\cdot \mathbf{{x}}_2 = {\mathbf{{p}}_2}'$.

\textbf{Step 1.3:} If $\mathbf{{x}}_1=\mathbf{{x}}_2$, compute $\mathbf{\tilde{y}}_{z,l,t}= \mu_{z,l} \cdot M_{l,t} $ as described in Algorithm \ref{alg:reg_err_free_mbr}. Otherwise, errors are detected in the help symbols. Exit the algorithm and switch to recovery regeneration mode.

\item [\step 2] No error has been detected for the calculating of the regeneration so far.
Let $\widetilde{Y}_z$ be a $q\times A$ matrix with the $l^{th}$ row defined as $[\mathbf{\tilde{y}}_{z,l,1}, \cdots, \mathbf{\tilde{y}}_{z,l,A/\alpha_l}], 0 \leq l \leq q-1$.



\item [\step 3] Calculate the regenerated symbols of the failed node $z$: $Y_{z'} = Y_z =  B_z \cdot \widetilde{Y}_z$.
\end{namelist}
\end{algorithms}

\begin{theorem}[H-MBR Regeneration--Detection Mode]
\label{th:reg_with_err_normal_mbr}

When the number of malicious nodes in the $d_l + 1$ helper nodes of Algorithm~\ref{alg:reg_with_err_normal_mbr} is less than $d_l + 1$, the probability for the bogus symbols sent from the malicious nodes to be detected is at least $1-1/q^2$.
\end{theorem}

\begin{proof}
Similar to the proof of Theorem~\ref{th:reg_with_err_normal}, we can write
\begin{eqnarray}
\mathbf{x}_1  &=& \mathbf{x} +
\mathbf{W}_{0,d_l-1,l}^{-1}  \cdot
[e_0, \cdots,e_{d_l-1} ]^T
  =\mathbf{x} +  \hat{\mathbf{x}}_1, \\
\mathbf{x}_2  &=& \mathbf{x} +
\mathbf{W}_{1,d_l,l}^{-1}  \cdot
[e_1, \cdots,e_{d_l} ]^T
  =\mathbf{x} +  \hat{\mathbf{x}}_2.
\end{eqnarray}

Since $\mathbf{W}_{0,d_l-1,l},\mathbf{W}_{1,d_l,l}$ are full rank matrices like the matrices $\mathbf{V}_{0,d_l-1,l},\mathbf{V}_{1,d_l,l}$ in the proof of Lemma~\ref{lm:reg_with_err_normal} and any $d_l$ vectors out of $\mu_{0,l},\mu_{1,l},\cdots,\mu_{d_l,l}$ are linearly independent, the rest of this proof is similar to that of Lemma~\ref{lm:reg_with_err_normal}. When the number of malicious nodes in the $d_l + 1$ helper nodes is less than $d_l + 1$, the probability for $\hat{\mathbf{x}}_1 = \hat{\mathbf{x}}_2$ is at most $1/q^2$. Therefore, the detection probability is at least $1-1/q^2$.
\end{proof}

\subsubsection{Recovery Mode}
Once the replacement node $z'$ detects errors using Algorithm~\ref{alg:reg_with_err_normal_mbr}, it will send integer $j = q-1$ to all the other $q^2 - 1$ nodes in the network requesting help symbols. Helper nodes will still use the way similar to that of the error-free network in Section~\ref{Sec:H-MBR-Regeneration-error-free} to send the help symbols. $z'$ can regenerate symbols using Algorithm~\ref{alg:reg_with_err_recovery_mbr}.

\begin{algorithms}[Recovery mode] $z'$ regenerates symbols of the failed node $z$ in hostile network
\label{alg:reg_with_err_recovery_mbr}

\normalfont
\begin{namelist}{\textbf{Step n:}}
\item [\step 1] For every $q-1 \geq l \geq 0$ in descending order and $1 \leq t \leq A/\alpha_{l}$ in ascending order, we can regenerate the symbols when the errors in the received help symbols ${\tilde{p}_{i,l,t}}'$ from $q^2 - 1$ helper nodes can be corrected. Without loss of generality, we assume $0 \leq i \leq q^2-2$.

\textbf{Step 1.1:} Let $\mathbf{{p}}'=[{\tilde{p}_{0,l,t}}', {\tilde{p}_{1,l,t}}', \cdots,  {\tilde{p}_{q^2-2,l,t}}']^T$. Since $\mathbf{W}_{0,q^2-2,l} \cdot \mathbf{{x}}  = \mathbf{p}'$,  $\mathbf{p}'$ can be viewed as an MDS code with parameters $(q^2-1,d_l,q^2-d_l)$.

\textbf{Step 1.2:} Substitute ${\tilde{p}_{i,l,t}}'$ in $\mathbf{{p}}'$ with the symbol $\otimes$ representing an erasure if node $i$ has been detected to be corrupted in the previous loops (previous values of $l,t$).

\textbf{Step 1.3:} If the number of erasures in $\mathbf{{p}}'$ is larger than $\min \{q^2- d_l - 1, \lfloor (q^2- d_{q-1} - 1)/2 \rfloor\}$, then the number of errors have exceeded the error correction capability. We will flag the decoding failure and exit the algorithm.

\textbf{Step 1.4:} Since the number of errors is within the error correction capability of the MDS code, decode $\mathbf{{p}}'$ to $\mathbf{p}_{cw}'$ and solve $\mathbf{{x}}$.

\textbf{Step 1.5:} If the $i^{th}$ position symbols of $\mathbf{{p}}'_{cw}$ and $\mathbf{{p}}'$ are different, mark node $i$ as corrupted.

\textbf{Step 1.6:} Compute $\mathbf{\tilde{y}}_{z,l,t}= \mu_{z,l} \cdot M_{l,t} $ as described in Algorithm \ref{alg:reg_err_free_mbr}.

\item [\step 2]
Let $\widetilde{Y}_z$ be a $q\times A$ matrix with the $l^{th}$ row defined as $[\mathbf{\tilde{y}}_{z,l,1}, \cdots, \mathbf{\tilde{y}}_{z,l,A/\alpha_l}], 0 \leq l \leq q-1$.



\item [\step 3] Calculate the regenerated symbols of the failed node $z$: $Y_{z'} = Y_z =  B_z \cdot \widetilde{Y}_z$.
\end{namelist}
\end{algorithms}

For data regeneration described in Algorithm~\ref{alg:reg_with_err_recovery_mbr}, since the structures of the underlying Hermitian codes of H-MSR code and H-MBR code with the same code rates are the same, we have similar result as that in Theorem~\ref{thm:reg_with_err_recovery}.
\begin{theorem}[H-MBR Regeneration--Recovery Mode]\label{thm:reg_with_err_recovery_mbr}
For data regeneration, the number of errors that the H-MBR code can correct is 
\begin{equation}\label{eq:num_of_errs_msr_recovery_mbr}
\tau_{H-MBR}  =  q \cdot \lfloor (q^2 - {d_{q-1}} -1)/2  \rfloor.
\end{equation}
\end{theorem}

\section{Reconstruction of the H-MBR code}
\label{Sec:H-MBR-Reconstruct}

In this section, we will first discuss reconstruction of the H-MBR code in error-free network. Then we will discuss reconstruction when there are corrupted nodes in the network.

\subsection{Reconstruction in Error-free Network}
\label{Sec:H-MBR-Reconstruct-error-free}

The main idea of the reconstruction algorithms is similar to that of the H-MSR code: reconstruct $g_l(\phi(^{s_i})),0\leq l \leq q-1$, by downloading help symbols from $k_l$ nodes, where $k_l$ represents the reconstruction parameter $k$ for $g_l(\phi(^{s_i}))$ in the H-MBR code.
We use Algorithm~\ref{alg:recon_err_free_mbr} in the network for the data collector DC to reconstruct the original file. For convenience, we suppose $k_q=0$.

Similar to the H-MSR code described in Section~\ref{Sec:H-MSR-Reconstruct-error-free}, DC will send requests to storage nodes. Upon receiving the request integer $j$, node $i$ will calculate and send symbols. When DC receives all the requested symbols, it can reconstruct the original file using the following algorithm:

\begin{algorithms}DC reconstructs the original file
\label{alg:recon_err_free_mbr}

\normalfont
\begin{namelist}{\textbf{Step n:}}
\item [\step 1] For every  $0 \leq l \leq q-1$, divide the symbol vector $\mathbf{\tilde{y}}_{i,l}$ into $A/\alpha_{l}$ equal row vectors: $[\mathbf{\tilde{y}}_{i,l,1},\mathbf{\tilde{y}}_{i,l,2}$, $\cdots,\mathbf{\tilde{y}}_{i,l,A/\alpha_{l}}]$.  ( $\mathbf{\tilde{y}}_{i,l}$ is the response from the $i^{th}$ node and we assume $0 \leq i \leq k_l-1$ without loss of generality.)

\item [\step 2] For every $0 \leq l \leq q-1$ and $1 \leq t \leq A/\alpha_{l}$, DC reconstructs the matrices related to the original file:

\textbf{Step 2.1:} Let $R=[\mathbf{\tilde{y}}_{0,l,t}^T, \mathbf{\tilde{y}}_{1,l,t}^T, \cdots, \mathbf{\tilde{y}}_{k_l-1,l,t}^T]^T$, we have the equation: $\mathbf{W}_{0,k_l-1,l} \cdot M_{l,t}  =  R$ according to the encoding algorithm.

\textbf{Step 2.2:} DC reconstructs $M_{l,t}$ using techniques similar to that of~\cite{Rashmi}.

\item [\step 3] DC reconstructs the original file from all the matrices $M_{l,t}$, $0 \leq l \leq q-1$ and $1 \leq t \leq A/\alpha_{l}$.
\end{namelist}
\end{algorithms}

%
%
%
%
%

\subsection{Reconstruction in Hostile Network}

Similar to the H-MSR code, the reconstruction algorithms for H-MBR code in error-free network do not work in hostile network. Even if the data collecter can calculate the symbol matrices $M$ using Algorithm~\ref{alg:recon_err_free_mbr}, it cannot verify whether the result is correct or not. There are two modes for the original file to be reconstructed in hostile network. One mode is the detection mode, in which no error has been found in the symbols received from the storage nodes. Once errors are detected in the detection mode, the recovery mode will be used to correct the errors and locate the malicious nodes.

\subsubsection{Detection Mode}

In the detection mode, DC will send requests in the way similar to that of the error-free network in Section~\ref{Sec:H-MBR-Reconstruct-error-free}. The only difference is that when $j=q-1$, DC will send requests to $k_{q-1} - k_q + 1$ nodes instead of $k_{q-1} - k_q$ nodes. Storage nodes will send symbols similar to that of the error-free network in Section~\ref{Sec:H-MBR-Reconstruct-error-free}. The reconstruction algorithm is described in Algorithm~\ref{alg:recon_with_err_normal_mbr} with the detection probability described in Theorem~\ref{th:recon_with_err_normal_mbr}.

\begin{algorithms}[Detection mode] DC reconstructs the original file in hostile network
\label{alg:recon_with_err_normal_mbr}

\normalfont
\begin{namelist}{\textbf{Step n:}}
\item [\step 1] For every  $0 \leq l \leq q-1$ , we can divide the symbol vector ${\mathbf{\tilde{y}}_{i,l}}'$ into $A/\alpha_{l}$ equal row vectors: $[{\mathbf{\tilde{y}}_{i,l,1}}',{\mathbf{\tilde{y}}_{i,l,2}}',\cdots,{\mathbf{\tilde{y}}_{i,l,A/\alpha_{l}}}']$. ${\mathbf{\tilde{y}}_{i,l}}'= \mathbf{\tilde{y}}_{i,l} + \mathbf{ e}_{i,l}$ is the response from the $i^{th}$ storage node. If $\mathbf{\tilde{y}}_{i,l}$ has been modified by the malicious node $i$, we have $\mathbf{e}_{i,l} \in (\mathbb{F}_{q^2})^{A} \backslash \{\mathbf{0}\}$. To detect whether there are errors, we will reconstruct the original file from two sets of storage nodes then compare the results. (Without loss of generality, we assume $0 \leq i \leq k_l$.)

\item [\step 2] For every $0 \leq l \leq q-1$ and $1 \leq t \leq A/\alpha_{l}$, DC can reconstruct the matrices related to the original file:

\textbf{Step 2.1:} Let
 ${R}'=[{\mathbf{\tilde{y}}_{0,l,t}}'^T, {\mathbf{\tilde{y}}_{1,l,t}}'^T, \cdots, {\mathbf{\tilde{y}}_{k_l,l,t}}'^T]^T$.

\textbf{Step 2.2:} Let
${R_1}' = [{\mathbf{\tilde{y}}_{0,l,t}}'^T, \cdots, {\mathbf{\tilde{y}}_{k_l-1,l,t}}'^T]^T$, which are the symbols collected from node $0$ to node $k_l-1$, then we have $\mathbf{W}_{0,k_l-1,l}\cdot M_1 ={R_1}'$. Solve $M_1$ using the method same to algorithm~\ref{alg:recon_err_free_mbr}.

\textbf{Step 2.3:} Let
${R_2}' = [{\mathbf{\tilde{y}}_{1,l,t}}'^T, \cdots, {\mathbf{\tilde{y}}_{k_l,l,t}}'^T]^T$, which are the symbols collected from node $1$ to node $k_l$, then we have $\mathbf{W}_{1,k_l,l} \cdot M_2= {R_2}'$. Solve $M_2$ using the method same to algorithm~\ref{alg:recon_err_free_mbr}.

\textbf{Step 2.4: } Compare $M_1$ with $M_2$. If they are the same, let $M_{l,t} = M_1$. Otherwise, errors are detected in the received symbols. Exit the algorithm and switch to recovery reconstruction mode.

\item [\step 3] No error has been detected for the calculating of the reconstruction so far.
So DC can reconstruct the original file from all the matrices $M_{l,t}$, $0 \leq l \leq q-1$ and $1 \leq t \leq A/\alpha_{l}$.
\end{namelist}
\end{algorithms}

\begin{theorem}[H-MBR Reconstruction--Detection Mode]
\label{th:recon_with_err_normal_mbr}

When the number of malicious nodes in the $k_l + 1$ nodes of Algorithm~\ref{alg:recon_with_err_normal_mbr} is less than $k_l + 1$, the probability for the bogus symbols sent from the malicious nodes to be detected is at least $1-1/q^{2\alpha_l}$.
\end{theorem}

\begin{proof}
For convenience, we write $\mathbf{e}_{i,l,t}$ as $\mathbf{e}_i$ in the proof. $\mathbf{e}_i \in {[\mathbb{F}_{q^2}]^{\alpha_l}}$ for $0 \leq i \leq k_l$. In Algorithm~\ref{alg:recon_with_err_normal_mbr}, ${R_1}'= R_1 + Q_1$
where $Q_1 = \begin{bmatrix}
\mathbf{e}_{0}\\
\mathbf{e}_{1}\\
\vdots \\
\mathbf{e}_{k_l-1}
\end{bmatrix}$.
Let $\mathbf{W}_{0,k_l-1,l} = [\Omega_{DC1},\Delta_{DC1}]$, $R_1 = [R_{1,1}, R_{1,2}]$ and $Q_1 = [Q_{1,1}, Q_{1,2}]$, where $\Omega_{DC1}$, $R_{1,1}$, $Q_{1,1}$ are $k_l \times k_l$ submatrices and $\Delta_{DC1}$, $R_{1,2}$, $Q_{1,2}$ are $k_l \times (\alpha_l - k_l)$ submatrices.

According to equation~(\ref{eq:mbr_m_matrix}), we have
\begin{eqnarray}
\mathbf{W}_{0,k_l-1,l}\cdot M_1 &=& [\Omega_{DC1}S_1 + \Delta_{DC1}T_1^T, \Omega_{DC1}T_1] \nonumber \\
&=& [R_{1,1} + Q_{1,1}, R_{1,2} + Q_{1,2}].
\end{eqnarray}
Since $\Omega_{DC1}$ is a submatrix of a Vandermonde matrix, it is a full rank matrix. So we have
\begin{eqnarray}
T_1 &=& \Omega_{DC1}^{-1} R_{1,2} + \Omega_{DC1}^{-1} Q_{1,2} = T + \hat{T}_1, \\
S_1 &=& \Omega_{DC1}^{-1} (R_{1,1} + Q_{1,1} - \Delta_{DC1}T_1^T) \nonumber \\
   &=& \Omega_{DC1}^{-1} (R_{1,1} - \Delta_{DC1}T^T) +
   \Omega_{DC1}^{-1} (Q_{1,1} - \Delta_{DC1} \hat{T}_1^T ) \nonumber \\
   &=& S + \Omega_{DC1}^{-1} (Q_{1,1} - \Delta_{DC1} \hat{T}_1^T) = S + \hat{S}_1.
\end{eqnarray}

For ${R_2}'= R_2 + Q_2$ in Algorithm~\ref{alg:recon_with_err_normal_mbr},
Let $R_2 = [R_{2,1}, R_{2,2}]$, $Q_2 = [Q_{2,1}, Q_{2,2}]$ and $\mathbf{W}_{1,k_l,l} = [\Omega_{DC2},\Delta_{DC2}]$, where $R_{2,1}$, $Q_{2,1}$, $\Omega_{DC2}$ are $k_l \times k_l$ submatrices and $R_{2,2}$, $Q_{2,2}$, $\Delta_{DC2}$ are $k_l  \times  (\alpha_l - k_l)$ submatrices. Similarly, we have
\begin{eqnarray}
T_2 &=& \Omega_{DC2}^{-1} R_{2,2} + \Omega_{DC2}^{-1} Q_{2,2} = T + \hat{T}_2, \\
S_2 &=& S + \Omega_{DC2}^{-1} (Q_{2,1} - \Delta_{DC2} \hat{T}_2^T) = S + \hat{S}_2.
\end{eqnarray}
If $\hat{T}_1 = \hat{T}_2$ and $\hat{S}_1 = \hat{S}_2$, Algorithm~\ref{alg:recon_with_err_normal_mbr} will fail to detect the bogus symbols. So we will focus on $\hat{T}_1,\hat{T}_2$ and $\hat{S}_1,\hat{S}_2$.

Suppose $\Pi_{1,j} = [e_0,\cdots,e_{k_l-1}]^T,\Pi_{2,j}=[e_1,\cdots,e_{k_l}]^T$ are the $j^{th},1 \leq j \leq \alpha_l-k_l,$ columns of $Q_{1,2}$ and $Q_{2,2}$ respectively, where $e_i \in \mathbb{F}_{q^2}$.
Since $\Omega_{DC1}$ and $\Omega_{DC2}$ are Vandermonde matrices and have the same relationship as that of between $\mathbf{V}_{0,d_l-1,l}$ and $\mathbf{V}_{1,d_l,l}$, similar as the proof of Lemma~\ref{lm:reg_with_err_normal}, we can prove that when the number of malicious nodes in the $k_l + 1$ nodes is less than $k_l + 1$, the probability of $\Omega_{DC1}^{-1}\Pi_{1,j} = \Omega_{DC2}^{-1}\Pi_{2,j}$ is at most $1/q^2$. Thus the probability for $\hat{T}_1 = \hat{T}_2$ is at most $1/q^{2(\alpha_l - k_l)}$. Through the same procedure, we can derive that the probability of $\hat{S}_1 = \hat{S}_2$ is at most $1/q^{2k_l}$. The probability for both $\hat{S}_1 = \hat{S}_2$ and $\hat{T}_1 = \hat{T}_2$ is at most $1/q^{2\alpha_l}$. So the detection probability is at least $1- 1/q^{2\alpha_l}$.
\end{proof}

\subsubsection{Recovery Mode}

Once DC detects errors using Algorithm~\ref{alg:recon_with_err_normal_mbr}, it will send integer $j = q-1$ to all the $q^2$ nodes in the network requesting symbols. Storage node $i$ will use the way similar to that of the error-free network in Section~\ref{Sec:H-MBR-Reconstruct-error-free} to send symbols. The reconstruct procedures are described in Algorithm~\ref{alg:recon_with_err_recovery_mbr}.

\begin{algorithms}[Recovery mode] DC reconstructs the original file in hostile network
\label{alg:recon_with_err_recovery_mbr}

\normalfont
\begin{namelist}{\textbf{Step n:}}
\item [\step 1] For every  $0 \leq l \leq q-1$ , divide the symbol vector ${\mathbf{\tilde{y}}_{i,l}}'$ into $A/\alpha_{l}$ equal row vectors: $[{\mathbf{\tilde{y}}_{i,l,1}}',{\mathbf{\tilde{y}}_{i,l,2}}',$ $\cdots,{\mathbf{\tilde{y}}_{i,l,A/\alpha_{l}}}']$. (Without loss of generality, we assume $0 \leq i \leq q^2-1$.)

\item [\step 2] For every $q-1 \geq l \geq 0$ in descending order and $1 \leq t \leq A/\alpha_{l}$ in ascending order, DC reconstructs the matrices related to the original file when the errors in the received symbol vectors ${\mathbf{\tilde{y}}_{i,l,t}}'$ from $q^2$ storage nodes can be corrected:

\textbf{Step 2.1:} Let
${R}'=[{\mathbf{\tilde{y}}_{0,l,t}}'^T, {\mathbf{\tilde{y}}_{1,l,t}}'^T, \cdots, {\mathbf{\tilde{y}}_{q^2-1,l,t}}'^T]^T$.

\textbf{Step 2.2:} If the number of corrupted nodes detected is larger than $\min \{ q^2- k_l , \lfloor (q^2 - k_{q-1})/2 \rfloor\}$, then the number of errors have exceeded the error correction capability. So here we will flag the decoding failure and exit the algorithm.

\textbf{Step 2.3:}  Since the number of errors is within the error correction capability of the H-MBR code, substitute ${\mathbf{\tilde{y}}_{i,l,t}}'$ in ${R}'$ with the symbol $\otimes$
representing an erasure vector if node $i$ has been detected to be corrupted in the previous loops (previous values of $l,t$).

\textbf{Step 2.4:} Solve $M_{l,t}$ using the method in section \ref{sec:Rec-M}. If symbols from node $i$ are detected to be erroneous during the calculation, mark node $i$ as corrupted.

\item [\step 3] DC reconstructs the original file from all the matrices $M_{l,t}$, $0 \leq l \leq q-1$ and $1 \leq t \leq A/\alpha_{l}$.
\end{namelist}
\end{algorithms}

For data reconstruction described in Algorithm~\ref{alg:recon_with_err_recovery_mbr}, since the structures of the underlying Hermitian codes of H-MSR code and H-MBR code with the same code rates are the same, we have similar result as that in Theorem~\ref{thm:recon_with_err_recovery}.
\begin{theorem}[H-MBR Reconstruction--Recovery Mode]\label{thm:recon_with_err_recovery_mbr}
For data reconstruction, the number of errors that the H-MBR code can correct is 
\begin{equation}\label{eq:num_of_errs_msr_recovery_recon_mbr}
\tau_{H-MBR}  =  q \cdot \lfloor (q^2 - {k_{q-1}})/2  \rfloor.
\end{equation}
\end{theorem}

\subsection{Recover Matrices $M_{\alpha_l,t}$ from $q^2$ Storage Nodes}\label{sec:Rec-M}

When there are bogus symbols ${\tilde{p}_{i,l,t}}'$ sent by the corrupted nodes for certain $l,t$, we can recover the matrices $M_{\alpha_l,t}$ as follows:

For ${R}'$ in Algorithm~\ref{alg:recon_with_err_recovery_mbr}, we have $\Phi_{DC} \cdot M'  =  {R}'$, where $\Phi_{DC} = \mathbf{W}_{0,q^2-1,l} = [\Omega_{DC}, \Delta_{DC}]$, $R' = [R_1',R_2']$. $\Omega_{DC}$, $R_1'$ are $q^2 \times k_l$ submatrices and $\Delta_{DC}$, $R_2'$ are $q^2 \times (\alpha_l - k_l)$ submatrices.

According to equation~(\ref{eq:mbr_m_matrix}), we have
\begin{equation}
\Phi_{DC} \cdot M' = [\Omega_{DC}S' + \Delta_{DC}T'^T, \Omega_{DC}T'] = [R_1', R_2'].
\end{equation}

For $R_2'=\Omega_{DC}T'$, we can view each column of $R_2'$ as a $(q^2, k_l, q^2 - k_l + 1 )$ MDS code because $\Phi_{DC}$ is a Vandermonde matrix. If the number of erasures $\sigma$ (corresponding to the previously detected corrupted nodes) and the number of corrupted nodes $\tau$ that have not been detected satisfy:
\begin{equation}
\sigma + 2\tau \leq q^2 - k_l,
\end{equation}
then all the columns of $T'$ can be recovered and the error locations (corresponding to the corrupted nodes) can be pinpointed. After $T'$ has been recovered, we can recover $S'$ following the same process because $\Omega_{DC}S' = R_1' - \Delta_{DC}T'^T$. So DC can reconstruct $M_{\alpha_l,t}$.

\section{Performance Analysis}\label{Sec:Performance}

In this section, we analyze the performance of the H-MSR code and compare it with the performance of the RS-MSR code. We will first analyze their error correction capability then their complexity.

The comparison results between the H-MBR code and the RS-MBR code are the same since the error correction capability and the complexity of the H-MSR code and the H-MBR code are similar while these performance parameters of the RS-MSR code and the RS-MBR code are similar.

\subsection{Scalable Error Correction}

\subsubsection{Error correction for data regeneration}

The RS-MSR code in~\cite{Rashmi-err} can correct up to $\tau$ errors by downloading symbols from $d+2\tau$ nodes. However, the number of errors may vary in the symbols sent by helper nodes. When there is no error or the number of errors is far less than $\tau$, downloading symbols from extra nodes will be a waste of bandwidth. When the number of errors is larger than $\tau$, the decoding process will fail without being detected. In this case, the symbols stored in the replacement node will be erroneous. If this erroneous node becomes a helper node later, the errors will propagate to other nodes.

The H-MSR code can detect the erroneous decodings using Algorithm~\ref{alg:reg_with_err_normal}. If no error is detected, regeneration of H-MSR only needs to download symbols from one more node than the regeneration in error-free network, while the extra cost for the RS-MSR code is $2\tau$.
If errors are detected in the symbols received from the helper nodes,
the H-MSR code can correct the errors using Algorithm~\ref{alg:reg_with_err_recovery}. Moreover,
the algorithm can determine whether the decoding is successful,
while the RS-MSR code is unable to provide such information.

\subsubsection{Error correction for data reconstruction}

The evaluation result is similar to the data regeneration. The RS-MSR code can correct up to $\tau$ errors with support from $2\tau$ additional helper nodes.
The H-MSR code is more flexible.  For error detection, it only requires symbols from one additional node using Algorithm~\ref{alg:recon_with_err_normal}.  The errors can then be corrected using Algorithm~\ref{alg:recon_with_err_recovery}.  The algorithm can also determine whether the decoding is successful.

\subsection{Error Correction Capability}


For data regeneration described in Algorithm~\ref{alg:reg_with_err_recovery}, according to Theorem~\ref{thm:reg_with_err_recovery} and equation~(\ref{eq:num_of_errs_msr_recovery}), the H-MSR code can correct $\tau_{H-MSR}=q \cdot \lfloor (q^2 - {d_{q-1}} -1)/2  \rfloor$ errors, while the $(q^3 - q, \sum_{l=0}^{q-1}{d_l},q^3 - q - \sum_{l=0}^{q-1}{d_l} + 1)$ RS-MSR code with the same rate can correct $\tau_{RS-MSR} = \lfloor (q^3 - q - \sum_{l=0}^{q-1}{d_l})/2 \rfloor$ errors. Therefore, we have the following theorem.
\begin{theorem}
For data regeneration, the number of errors that the H-MSR code and the RS-MSR code can correct satisfy $\tau_{H-MSR} > \tau_{RS-MSR}$ when $q \geq 3$.
\end{theorem}

\begin{proof}
For $\tau_{RS-MSR}$, we have
%
\begin{eqnarray}
\label{eq:min_dist_rs}
\tau_{RS-MSR} &=& \left\lfloor \left(q^3 - q - \sum_{l=0}^{q-1}{d_l}\right)/2 \right\rfloor \\
       & \leq &  \left\lfloor (q^3 - q - q\cdot d_{q-1} - \frac q2(q-1) )/2  \right\rfloor \nonumber \\
       & = &  \left\lfloor q \cdot (q^2 - {d_{q-1}} -1)/2 -\frac {q(q-1)}4  \right\rfloor \nonumber\\
       & \leq &   q \cdot (q^2 - {d_{q-1}} -1)/2 -\frac {q(q-1)}4  . \nonumber
\end{eqnarray}

For $\tau_{H-MSR}$, we have
\begin{equation}
\label{eq:min_dist_h}
\tau_{H-MSR}  =  q \cdot \lfloor (q^2 - {d_{q-1}} -1)/2  \rfloor.
\end{equation}

When $q=3$, it is easy to verify that $ \tau_{H-MSR} > \tau_{RS-MSR}$.

When $q>3$, We can rewrite equation~(\ref{eq:min_dist_h}) as
\begin{equation}
\tau_{H-MSR}  \geq  q \cdot(q^2 - {d_{q-1}} -1)/2   - q/2.
\end{equation}
The gap between $\tau_{H-MSR}$ and $\tau_{RS-MSR}$ is at least
\begin{equation}
\label{eq:gap}
\frac {q(q-1)}4 - \frac {q}2 =   \frac {q^2 - 3q}4 >0 \, , q>3 ,
\end{equation}
so we have $ \tau_{H-MSR} > \tau_{RS-MSR}$.
\end{proof}

\begin{example}
Suppose $q=4$ and $m=37$, the Hermitian curve is defined by $y^4 + y = x^5$ over $\mathbb{F}_{4^2}$. From the previous discussion, we have $\kappa(0)=10,\kappa(1)=9,\kappa(2)=7,\kappa(3)=6$. Choose $\alpha_0=6,\alpha_1=5,\alpha_2=4,\alpha_3=3$. So $d_0=12,d_1=10,d_2=8,d_3=6$. According to the analysis above, we have $\tau_{H-MSR} = 4 \cdot \tau_3 = 4 \cdot \lfloor (15-6)/2 \rfloor = 16$, which is larger than $\tau_{RS-MSR} = \lfloor (60-36)/2 \rfloor = 12 $.
\end{example}

We also show that the maximum number of malicious nodes from which the errors can be corrected by the H-MSR code in Fig.~\ref{fig:comparison}.  Here the parameter $q$ of the Hermitian code increases from 4 to 16 with a step of 2. In the figure, the code rates for the RS-MSR code and the H-MSR code are the same. The figure demonstrates that for data regeneration, the H-MSR code has better error correction capability than the RS-MSR code.

\begin{figure}
\centering
\includegraphics[width=3.5in]{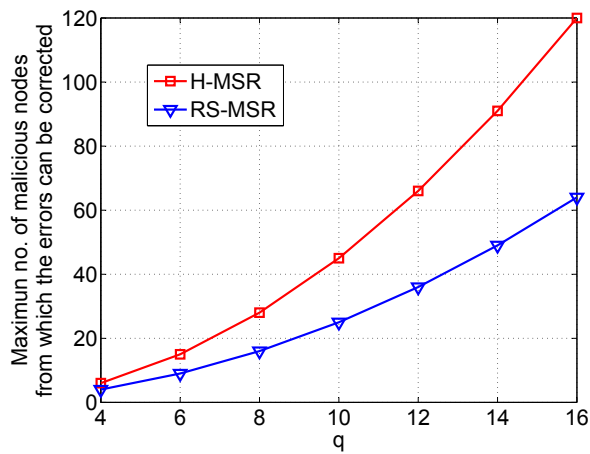}
\caption{Comparison of error correction capability between the H-MSR code and the RS-MSR code}
\label{fig:comparison}
\end{figure}


For data reconstruction described in Algorithm~\ref{alg:recon_with_err_recovery}, according to Theorem~\ref{thm:recon_with_err_recovery} and equation~(\ref{eq:num_of_errs_msr_recovery_recon}), the number of errors that H-MSR code can correct is $\tau_{H-MSR}  =  q \cdot \lfloor (q^2 - {k_{q-1}})/2  \rfloor$.
Similarly, we can conclude that for data reconstruction the H-MSR code has better error correction capability than the RS-MSR code under the same code rate.

\subsection{Complexity Discussion}

For the complexity of the H-MSR code, we consider two scenarios.

\subsubsection{H-MSR regeneration}

For the H-MSR regeneration, compared with RS-MSR code, the H-MSR code will slightly increase the complexity of the helper nodes. For each helper node, the extra operation is a matrix multiplication between $B_i^{-1}$ and $Y_i$. The  complexity is $O(q^2)=\calO((n^{1/3})^2)=\calO(n^{2/3})$.
Similar to \cite{Hermitian}, for a replacement node, from Algorithm~\ref{alg:reg_err_free} and Algorithm~\ref{alg:reg_with_err_normal}, we can derive that the complexity to regenerate symbols for RS-MSR
is $\calO(n^2)$, while the complexity for H-MSR is only $\calO(n^{5/3})$.
Likewise, for Algorithm \ref{alg:reg_with_err_recovery}, the complexity to recover the H-MSR code is $\calO(n^{5/3})$, and $\calO(n^2)$ for RS-MSR code.

\subsubsection{H-MSR reconstruction}

For the reconstruction, compared with RS-MSR code, 
the additional complexity of the H-MSR code for each storage node is $\calO(q^2)$, which is
$\calO(n^{2/3})$.  The computational complexity for DC to reconstruct the data is $\calO(n^{5/3})$ for the H-MSR code and $\calO(n^2)$ for the RS-MSR code.

\section{Conclusion} \label{Sec:Conclusion}

In this paper, we developed a Hermitian code based minimum storage regeneration (H-MSR) code and a Hermitian code based minimum bandwidth regeneration (H-MBR) code for distributed storage.  Due to the structure of Hermitian code, our proposed codes can significantly improve the performance of the regenerating code under malicious attacks.  In particular, these codes can deal with errors beyond the maximum distance separable (MDS) code. Our theoretical analyses demonstrate that the H-MSR/H-MBR codes have lower complexity than the Reed-Solomon based minimum storage regeneration (RS-MSR) code and the Reed-Solomon based minimum bandwidth regeneration (RS-MBR) code in both regeneration and reconstruction. As a future research task, we will further analyze the optimal design of regenerating code based on the Hermitian-like codes.

\bibliographystyle{ieeetr} %

\end{document}